\documentclass[prd]{revtex4}
\usepackage{graphicx,xcolor,soul,tikz}
\usetikzlibrary{patterns}
\usepackage{amsmath,amssymb,amsthm}
\newtheorem{theorem}{Theorem}

\newcommand{\RR}{{\mathbb R}}

\newcommand{\NN}{{\mathbb N}}

\newcommand{\fhat}{\hat{f}}
\newcommand{\ghat}{\hat{g}}

\newcommand{\kb}{{\boldsymbol{k}}}
\newcommand{\mb}{{\boldsymbol{m}}}
\newcommand{\qb}{{\boldsymbol{q}}}
\newcommand{\qbh}{\hat{\boldsymbol{q}}}
\newcommand{\lb}{{\boldsymbol{\ell}}}
\newcommand{\xb}{{\boldsymbol{x}}}
\newcommand{\yb}{{\boldsymbol{y}}}

\begin{document}
\baselineskip=16pt

\title{Probability Distributions for Space and Time Averaged Quantum Stress Tensors }

\author{Christopher J. Fewster}
\email{chris.fewster@york.ac.uk}
\affiliation{Department of Mathematics, University of York, 
Heslington, York YO10 5DD,
United Kingdom}

\author{L. H. Ford}
\email{ford@cosmos.phy.tufts.edu}
\affiliation{Institute of Cosmology, Department of Physics and
  Astronomy, 
Tufts University, Medford, Massachusetts 02155, USA}

\begin{abstract}
We extend previous work on quantum stress tensor operators which have been averaged over finite time intervals to include averaging
over finite regions of space as well. The space and time averaging can be viewed as describing a measurement process for a stress
tensor component, such as the energy density of a quantized field in its vacuum state. Although spatial averaging reduces the probability
of large vacuum fluctuations compared to time averaging alone, we find that the probability distribution decreases more slowly than exponentially
as the magnitude of the measured energy density increases. This implies that vacuum fluctuations can sometimes dominate over thermal
fluctuations and potentially have observable effects.
\end{abstract}

\maketitle

\section{Introduction}
\label{sec:intro}

Although the vacuum state of a quantum field theory is an eigenstate of the Hamiltonian, the integral of the energy density over all space, it 
is not an eigenstate of the local energy density or of other components of the stress tensor. This implies the existence of vacuum fluctuations 
of the energy density and other quadratic operators. For these fluctuations to be finite, and hence physically meaningful, these operators
must be averaged over a finite spacetime region. We can view the averaging process as representing the outcome of a measurement of the
operator. The energy density at a single spacetime point is not measurable, and hence not meaningful. However, the spacetime average is
meaningful, and will have finite fluctuations described by a probability distribution. 

The study of the probability distributions for quantum stress tensors was begun in Ref.~\cite{FewsterFordRoman:2010} for conformal field theory (CFT) in
two spacetime dimensions, and continued in Refs.~\cite{FFR2012}  and \cite{FF2015} for quantum fields in flat four dimensional spacetime. 
Further results on CFT appear in~\cite{Fe&Ho18}. Let $x$
denote a dimensionless measure of the averaged stress tensor component $T$. If $\tau$ is a measure of the size of the sampling region, then in
units where $\hbar=c=1$, we may take $x= \tau^d\, T$, where $d$ is the dimension of the spacetime. Let $P(x)$ denote a probability distribution
so that $P(x)\, dx$ is the probability in a measurement of finding an outcome in the interval $[x,x+dx]$. There are two key features of $P(x)$ for 
a quadratic operator, such as the energy density, which
have emerged in the papers just cited: 1) There is a negative lower bound on the region where $P(x) \not=0$ if $T \geq 0$ at the classical level, and
2)  $P(x)$ can fall more slowly than exponentially, leading to an enhanced probability for large positive fluctuations relative to thermal fluctuations.
By contrast, the probability distribution for the spacetime average of a linear operator, such as the electric field, is a Gaussian function.

If $T$ is a non-negative quantity in classical physics, such as the energy density, its quantization typically admits quantum states for which its expectation value 
is below the vacuum value. In particular, if the vacuum expectation value vanishes there exist states for which
its expectation value is negative, $\langle T  \rangle < 0$, so regions where the mean energy density is negative become possible. At least in some models, these regions are 
constrained by quantum inequalities of the form $\langle  \tau^d\,T  \rangle \geq - x_0$, where $x_0 >0$ is a dimensionless number of the order of 
or somewhat less than unity. For a recent review see~\cite{F2017}. If the quantum inequality gives the optimal lower bound on expectation values, 
then  $P(x) = 0$ if $x < -x_0$. This means
that $-x_0$ is the lowest eigenvalue of the averaged operator $T$, and is hence both the lower bound on expectation values, and the smallest possible
outcome of a measurement in any state.

For the energy density (at least for the averages considered to date) the tail of  $P(x)$ for $x \gg 1$ was found to fall as an exponential in two spacetime 
dimensions~\cite{FewsterFordRoman:2010,Fe&Ho18}, 
but more slowly in four dimensions~\cite{FFR2012,FF2015}.  Specifically, $P(x) \sim c_0\, x^b \, {\rm e}^{-a x^c}$ for some constants 
$c_0, b, a, c$, of which $c$ is the most crucial.  For stress tensor operators averaged in time with a Lorentzian function, it was found in  Ref.~\cite{FFR2012}
that $c=1/3$. This implies that the distribution is highly skewed and so fluctuations which are several orders of magnitude larger than the standard deviation 
can have a non-negligible probability
of occurring. This is a result which would not be possible in random processes where measurements at different moments in time are uncorrelated, in which case 
the central limit theorem would give a Gaussian probability distribution. By contrast our results reflect the highly correlated nature of quantum vacuum fluctuations. 

Although a Lorentzian function of time is a useful model, it suffers from the defect that it describes a measurement which began in the infinite past and
is only completed in the infinite future. A more realistic description involves smooth (infinitely differentiable) functions which have compact support, that is,
are zero outside of a finite interval. The  probability distributions for quantum stress tensors measured in a finite interval with such functions was studied
in Ref.~\cite{FF2015}. A class of compactly supported functions was treated, whose Fourier transforms fall as ${\rm e}^{-\gamma|\omega|^\alpha}$, where
$0< \alpha <1$ and $\gamma>0$, as $|\omega|\rightarrow \infty$. It was argued that such functions could arise in physical situations, as illustrated by a 
simple electrical circuit whose switch-on corresponds to $\alpha = 1/2$. For this class of functions, it was
shown that the tail of the probability distribution now decays with $c = \alpha/3$. Thus if, for example, a measurement of the energy density in the 
vacuum state of the electromagnetic field is described by the $\alpha = 1/2$ function, then the probability of finding a very large energy density associated with
$x \gg 1$ will be roughly proportional to ${\rm e}^{-a x^{1/6}}$. 

The previous results on stress tensor probability distributions~\cite{FewsterFordRoman:2010,FFR2012,FF2015} were obtained either from a moment generating 
function~\cite{FewsterFordRoman:2010}, or by asymptotic calculation of high moments~\cite{FFR2012,FF2015}. In four dimensions, the moments
approach suffers from the ambiguity that the moments do not necessarily uniquely determine $P(x)$. The Hamburger moment theorem~\cite{Simon} guarantees 
that $P(x)$ is uniquely determined by the moments of the operator provided that the $n$-moment grows no faster than $n!D^n$ as $n \rightarrow \infty$, for some constant $D$. 
However, the moments of stress tensor operators averaged with the compactly supported functions of time discussed in Ref.~\cite{FF2015} grow as
$(3 n/\alpha)!$. The non-compactly supported Lorentzian function used in Ref.~\cite{FFR2012} formally corresponds to the $\alpha = 1$ case, and
leads to moments with $(3 n)!$ growth. In all of these cases, $P(x)$ may not be uniquely determined from the moments. In general, when the moments 
grow too rapidly to ensure uniqueness,
there can be several distinct choices for $P(x)$ which all produce the same moments, and differ from one another by an oscillatory function of $x$.
Even if $P(x)$ is not uniquely determined, its integrals over a finite interval tend to cancel the oscillations and can give a reliable estimate of the probability of
a result in this interval. For example, in some applications one is interested in the probability of a fluctuation which exceeds a given threshold and is
given by the complementary cumulative distribution, $P_>(x) = \int_x^\infty P(y) \, dy$, and it is possible to extract bounds on this function from the moment sequence in some cases, even if the moment sequence does not determine the probability distribution uniquely~\cite{FFR2012}.

There is also an independent approach to finding $P(x)$ which does not use the moments, which is direct diagonalization of the averaged operator
$T$ by a Bogoliubov transformation to find its eigenvalues and eigenstates. The probability of finding a given eigenvalue in a measurement on the
original vacuum state is then the squared overlap of the eigenstate with the vacuum. In practice, this approach must be performed numerically on a
system with a finite number of degrees of freedom. This was done in Ref.~\cite{SFF18} for a massless scalar field in a spherical cavity including about
one hundred modes for time sampling associated with several values of $\alpha$. The results are in reasonable agreement with those found for the
tail of $P(x)$ in Refs.~\cite{FFR2012,FF2015}. This lends support to the conclusion in the latter references that fluctuations several orders of magnitude
larger than the the typical  fluctuation can have a non-negligible probability of occurrence.

Such large fluctuations may have potentially observable effects. For example, the role of large radiation pressure fluctuations in enhancing the barrier
penetration by charged particles was treated in Ref.~\cite{Huang:2016kmx}, where it was argued that these fluctuations have the potential in some
circumstances to increase the barrier penetration rate by several orders of magnitude compare to the rate predicted by the usual quantum tunneling
process. It was further suggested that this effect may have already been observed in the nuclear fusion of heavy ions with heavy nuclei. By contrast, the  
vacuum fluctuations of the linear electric field, which obey a Gaussian probability distribution, cause only a modest increase in penetration 
rates~\cite{FZ99,Huang:2015lea}.  Quantum stress tensor fluctuations are also of interest in gravity theory, as they can drive passive fluctuations of
the gravitational field, which is a variety of quantum gravity effect. Stress tensor fluctuations in the early universe could play a role in the creation of
primordial density perturbations~\cite{WKF07,Ford:2010wd} or tensor perturbations~\cite{Wu:2011gk}. The references just cited all deal with integrals 
of the stress tensor correlation function, and hence the variance of the stress tensor fluctuations. It will be of interest to study the probability of large
fluctuations in these and other gravitational applications.  One possible application is to the effects of vacuum fluctuations on the small scale causal structure
of spacetime. In two-dimensional models, it has been found that large positive fluctuations can cause focussing of geodesics,
and closure of lightcones on small scales~\cite{CMP11,CMP18}.

Most of the previous work on the probability of  quantum stress tensors fluctuations was restricted to operators averaged in time at one spatial point. The 
purpose of the  present paper is to extend this treatment to include the effects of averaging in space as well. The outline of the paper is as follows:
In Sec.~\ref{sec:CFT}, we discuss stress tensor probability distributions in two spacetime dimensions, particularly in conformal field theory where exact
results are possible. Space and time averaging of stress tensor operators in four-dimensional Minkowski spacetime is developed in Sec.~\ref{sec:spatial-ave},
and the sampling functions needed for this averaging are discussed.
An iteration procedure for the calculation of the moments of the averaged operators is introduced. This procedure is analyzed in detail in 
Sec.~\ref{sec:iteration}. It is argued that if the spatial averaging scale is smaller than the temporal scale, then the lower moments are sensitive only to
the time averaging, but the high moments will also depend upon spatial averaging. The implications of these results for the rate of growth of the
moments is treated in Sec.~\ref{sec:growth}. It is found that the initial growth rate can be the  $(3 n/\alpha)!$ behavior found in Ref.~\cite{FF2015}
with time averaging alone. However, for larger $n$, there is a transition to a somewhat lower growth rate of $(n/\alpha)!$. This is still too fast to satisfy the 
Hamburger criterion, but our results suggest that a weaker criterion due to Stieltjes holds for $1/2\le\alpha<1$, implying that the moments uniquely determine 
the probability distribution among those that vanish on a half-line. The implications of these
results for the tail of the  probability distribution are discussed in Sec.~\ref{sec:tail}, where it is shown that the asymptotic form of $P(x)$ now falls
more rapidly than in the worldline case, but still more slowly than an exponential function. This reflects that fact that spatial averaging somewhat reduces
the probability of large fluctuations, but this probability remains high enough to have important physical effects. The latter point is discussed in more detail
in the final section, Sec.~\ref{sec:final}, where the key results of the paper are summarized and discussed. Appendix~A contains an
explicit construction of specific forms of the temporal and spatial sampling functions. Appendix~B discusses some results on the
asymptotic forms of integrals which are used in Sec.~\ref{sec:growth}.

Units in which $\hbar = c =1$ are used throughout the paper.

\section{Exact results in 2-dimensional conformal field theory}
\label{sec:CFT}

Two-dimensional conformal field theory (CFT) provides an interesting example, in which the relative effects of time and space averaging can be determined in detail. 
Recall that the energy density of a CFT in $1+1$-dimensions splits into mutually commuting left- and right-moving components
\begin{equation}
T_{00}(t,x) = T_R(u) + T_L(v),
\end{equation}
where we assume flat spacetime and let $u=t-x$, $v=t+x$. Any spacetime average of the energy density can be written in terms of these components as
\begin{equation}
\int dx\,dt\, T_{00}(t,x)f(t,x) = \int dv\, T_L(v) F_L(v) + \int du\, T_R(u) F_R(u) ,
\end{equation}
where
\begin{align*}
F_L(v) &= \frac{1}{2}\int_{-\infty}^\infty du\, f\left(\frac{u+v}{2},\frac{v-u}{2}\right) \\
F_R(u) &= \frac{1}{2}\int_{-\infty}^\infty dv\, f\left(\frac{u+v}{2},\frac{v-u}{2}\right).
\end{align*}
Here, the leading factor of $1/2$ is a Jacobian determinant. Now let $P_L$ be the probability density function for measurements of $T_L$, averaged against $F_L$, in the vacuum state, i.e.,
\begin{equation}
\int_{\omega_1}^{\omega_2} d\omega\,P_L(\omega) = \text{Prob} \Bigl(T_L(F_L)\in[\omega_1,\omega_2] \Bigr)
\end{equation}
and write $P_R$ and $P$ for the analogous probability density functions of $T_R$ (averaged against $F_R$) and $T_{00}$ (averaged against $f$). As $T_L$ and $T_R$ commute, 
the probability distributions are independent and the combined probability distribution is obtained as their convolution,
\begin{equation}
P(\lambda) = \int_{-\infty}^\infty d\lambda' P_L(\lambda-\lambda')P_R(\lambda').
\end{equation}
The probability distribution of these components of the energy tensor can be determined -- at least in principle -- either by a moment generating function 
method~\cite{FewsterFordRoman:2010} or by conformal welding techniques~\cite{Fe&Ho18}. The latter method can be applied to the cases of 
the vacuum and certain other special states, including thermal equilibrium states and also highest weight states~\cite{Fe&Ho18}. Each method rests on the solution to certain subsidiary problems and closed form results 
are only available in particular cases~\cite{FewsterFordRoman:2010,Fe&Ho18,AF19}, though the method of~\cite{Fe&Ho18} is also amenable to numerical treatment. 

Here, we draw attention to a special case where the probability distribution can be determined in closed form for different spatial and temporal averaging scales. Let 
\begin{equation}
f(t,x)= (2\pi \ell \tau)^{-1} e^{-\frac{1}{2}(t^2/\tau^2+x^2/\ell^2)},
\end{equation}
that is, a product of Gaussians in space and time, normalized to have unit integral over spacetime, in which $\ell$ and $\tau$  determine the spatial and temporal averaging scales.
In this case, a simple calculation gives
\begin{equation}
F_L(u) = \frac{e^{-u^2/(2\sigma^2)}}{\sqrt{2\pi\sigma^2}},
\end{equation} 
which is also a normalized Gaussian with characteristic width $\sigma=\sqrt{\ell^2+\tau^2}$. 
It is easily seen that $F_R(v)=F_L(v)$. For any unitary positive energy CFT, the probability distribution of $T_L(F_L)$ in the vacuum state is known in closed form~\cite{FewsterFordRoman:2010} 
(see~\cite{Fe&Ho18,AF19} for some other closed form expressions) and is given by 
the shifted Gamma distribution
\begin{equation}
P_L(\omega) = \vartheta(\omega+\omega_0)\frac{(2\pi\sigma^2)^{c/24}}{\Gamma(c/24)}(\omega+\omega_0)^{c/24-1}e^{-2\pi\sigma^2(\omega+\omega_0)},
\end{equation}
where $c$ is the central charge of the CFT [e.g., $c=1$ for a massless scalar field], $\omega_0=c/(48\pi\sigma^2)$ and $\vartheta$ is a Heaviside function. 
As $P_L$ and $P_R$ are 
identical, the overall probability distribution is the convolution of $P_L$ with itself and is again a shifted Gamma distribution 
\begin{equation}\label{eq:rho}
P(\omega)= \vartheta(\omega+2\omega_0)\frac{(2\pi\sigma^2)^{c/12}}{\Gamma(c/12)}(\omega+2\omega_0)^{c/12-1}e^{-2\pi\sigma^2(\omega+2\omega_0)}.
\end{equation}
To see this, it is easiest to proceed from the moment generating function 
\begin{equation}
M_L(\mu) =\int_{-\infty}^\infty d\mu\, P_L(\omega) e^{\mu\omega} = \left[\frac{e^{-\mu/(2\pi\sigma^2)}}{1-\mu/(2\pi\sigma^2)}\right]^{c/24}
\end{equation}
for $P_L$ (defined for $\mu<2\pi\sigma^2$) and note that the moment generating function for $P$ must be 
\begin{equation}\label{eq:CFTMGF}
M(\mu) = M_L(\mu)^2 =  \left[\frac{e^{-\mu/(2\pi\sigma^2)}}{1-\mu/(2\pi\sigma^2)}\right]^{c/12}.
\end{equation}
Therefore the probability density function for $P$ is just that of $P_L$ but with $c$ replaced by $2c$ throughout. 

We may read off a sharp quantum inequality bound on the averaged energy density from \eqref{eq:rho}, namely
\begin{equation}\label{eq:sharpQEI}
\int dt\,dx\, \langle T_{00}(t,x)\rangle_\psi f(t,x) \ge -\frac{c}{24\pi(\ell^2+\tau^2)}
\end{equation}
for any physically reasonable state $\psi$. This inequality may also be obtained as a special case of a general quantum inequality bound proved by different methods in~\cite{Fe&Ho05}, in which a precise specification of the relevant states may be found. It is interesting to compare this bound with the worldline bound 
\begin{equation}
\frac{1}{\tau\sqrt{2\pi}} \int dt\,    e^{-t^2/(2\tau^2)}\langle T_{00}(t,x)\rangle_\psi  \ge -\frac{c}{24\pi \tau^2}
\end{equation}
obtained in~\cite{FewsterFordRoman:2010,Fe&Ho05} for Gaussian smearing on timescale $\tau$. If one attempted to derive a spacetime bound by simply averaging all these bounds in 
$x$ with the appropriate Gaussian weight, one would obtain a (non-sharp) bound 
\begin{equation}\label{eq:nonsharpQEI}
\int dt\,dx\, \langle T_{00}(t,x)\rangle_\psi f(t,x) \ge -\frac{c}{24\pi \tau^2}.
\end{equation}
As one might expect, the sharp bound  \eqref{eq:sharpQEI} improves on this for all $\ell >0$, and becomes progressively tighter as $\ell$ increases. 
In the limit $\ell\to\infty$, we see that 
the sharp lower bound in \eqref{eq:sharpQEI} vanishes, which is to be expected as the Hamiltonian is a positive operator. Similarly, the probability distribution 
\eqref{eq:rho} converges 
to the delta-distribution $\delta(\omega)$ in this limit, reflecting the fact that vacuum measurements of the Hamiltonian result in $0$ with probability $1$. 

Our main interest, however, is in the effect of the spatial averaging on the moments and the probability distribution for finite spatial averaging scales. Inspecting the 
moment generating  function \eqref{eq:CFTMGF}, it is clear that the $n$-th moment scales with the characteristic scale $\sigma$ as
\begin{equation}
M_n^{(\tau,s)} = \left(\frac{\tau^2}{\sigma^2}\right)^n M_n^{(\tau,0)} = \left(1+(\ell/\tau)^2\right)^{-n} M_n^{(\tau,0)}.
\end{equation}
For $n(\ell/\tau)^2\ll 1$, the moments are little changed from those obtained by pure worldline smearing. 
This is a special case of a more general effect whereby a worldline result can be obtained as a limit of a small spatial averaging scale, which will be discussed in 
Sec.~\ref{sec:transition}. At higher $n$, of course, the effects of the spatial averaging become apparent. Likewise, for a range of values $\omega$ slightly greater
than zero, the probability distribution of $\rho$ is well-approximated by its values 
for $\ell=0$ (with $\tau$ fixed), but as $\omega$ increases, the two distributions depart from one another, with the $\ell>0$ distribution decaying exponentially faster. 
An illustrative plot appears 
in Fig.~\ref{fig:probdist}. Note, however, that the probability of finding a negative measurement outcome is
given in terms of the lower incomplete $\Gamma$-function as 
\begin{equation}
\text{Prob}(T_{00}(F)\le 0) = \frac{\gamma(c/12,c/12)}{\Gamma(c/12)},
\end{equation}
which is independent of $s$ and $\tau$, and depends only on the 
central charge $c$  (provided we maintain Gaussian sampling). Some results for non-Gaussian worldline sampling can be found in Ref.~\cite{Fe&Ho18,AF19}.

Extrapolating from these results, we may expect that for general quantum field theories, spatial averaging reduces the magnitude of the quantum inequality bound and 
also causes the positive  tail of the probability distribution to decay more rapidly. Nonetheless, we may also expect that for sufficiently low moments or for a range of 
smaller values in the probability distribution, one  may neglect the effect of spatial averaging on scales small in relation to the temporal averaging. Nonetheless, not all 
features of the CFT might be expected to generalize. In particular,  here the spacetime averaged probability distribution is of the same functional form as the worldline 
averaged case, but with different parameters. As we will see, this is a special feature of  conformal fields and is not true in general. 

\begin{figure}[t]
	\begin{center}
		\begin{tikzpicture}[yscale=1,xscale=1.5]
		\draw[very thin, ->] (-2,0) -- (5,0) node[right] {$\omega$};
		\draw[very thin, ->] (0,-0.5) -- (0,7) node[right] {$P(\omega)$} ;
		\draw[dashed] (-15/24,-0.5) -- ++(0,7.5);
		\draw[dashed] (-15/120,-0.5) -- ++(0,7.5);
		\draw[smooth,domain=-0.84/24:0.3,line width=1pt,color=red] plot (15*\x,{(\x+1/24)^(-11/12)*exp(-2*\x-1/24)/15});
		\draw[line width =1.5pt,color=red] (-15/24,0.05) -- (-2,0.05);
		\draw[line width =1.5pt,color=blue] (-15/120,-0.05) -- (-2,-0.05);
		\draw[smooth,domain=-4/3600:0.3,line width=1pt,color=blue] plot (15*\x,{(\x+1/120)^(-11/12)*5^(1/12)*exp(-10*\x-1/24)/15});
		\node[below right] at (0,0) {$0$};
		\end{tikzpicture}
		\caption{The probability density $P(\omega)$ plotted for central charge $c=1$ with averaging along a worldline (left-hand curve, red) and for spacetime averaging 
		with the same temporal sampling scale $\tau$ and $\ell=2\tau$ (right-hand curve, blue). The latter is displaced to the right and decays more rapidly. The vertical 
		asymptotes occur at the quantum inequality bound in each case.}
		\label{fig:probdist}
	\end{center}
\end{figure}
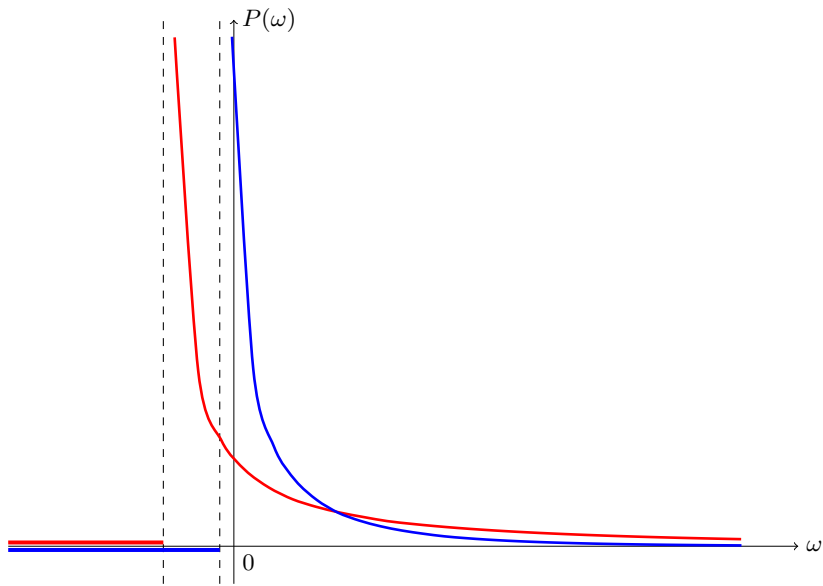

\section{Moments with Spatial Averaging}
\label{sec:spatial-ave}

\subsection{Averaged operators and their moments}

Let ${\cal T}(\xb,t)$ be a quadratic normal ordered bosonic operator in four dimensional flat spacetime, such as a stress tensor component for a free scalar or electromagnetic field. 
We consider a space and time average of this operator defined by 
\begin{equation}
T = \int_{-\infty}^\infty dt \, f(t) \, \int d^3 x \, g(\xb) \; {\cal T}(\xb,t)\,,
\end{equation} 
where $f(t)$ and $g(\xb)$ are compactly supported functions of time and of space, respectively. They are assumed to be non-negative and satisfy
 \begin{equation}
 \int_{-\infty}^\infty dt \, f(t) = 1\,,
 \label{eq:fnorm}
\end{equation} 
and 
\begin{equation}
\int d^3 x \, g(\xb) = 1.
 \label{eq:gnorm}
\end{equation} 

 Note that the averaging process breaks Lorentz symmetry. This is to be expected, as the averaging describes a measurement made in a specific spacetime region
 and in a selected frame of reference. 
The space and time averaged operator may be expanded in terms of annihilation and creation operators in the form
\begin{equation}
T = \sum_{i\, j} (A_{i j}\, a^\dagger_i \,a_j + B_{i j}\, a_i \,a_j + 
B^*_{i j} \, a^\dagger_i \,a^\dagger_j ) \,,
\label{eq:T}
\end{equation}
where $[a_i,a_j^\dagger]=\delta_{ij}\openone$, $A$ is hermitian and $B$ is symmetric.
The moments of $T$ are  defined as the vacuum expectation values of powers of $T$:
\begin{equation}
\mu_n = \langle T^n \rangle \,.
\end{equation}
The various moments can be expressed as polynomials in the matrices, $A_{i j}$ and $B_{i j}$. 
The second moment, for example, is given by
\begin{equation}
\mu_2 = 2\, \text{Tr}\, B^\dagger B = 2\sum_{j \ell} |B_{j \ell}|^2\,. 
\label{eq:mu2}
\end{equation}

The primary example which we investigate in this paper is ${\cal T} = :\dot{\varphi}^2:$, the squared time derivative of a massless scalar
field. We may write a mode expansion for $\dot{\varphi}$ as  
 \begin{equation}
\dot{\varphi}(t,\xb) = \sum_\kb\sqrt{ \frac{\omega}{2 V}}\, \left( a_\kb\, {\rm e}^{i (\kb\cdot \xb -\omega t)}
 + a^\dagger_\kb\, {\rm e}^{-i (\kb \cdot \xb -\omega t)} \right)\,,
\end{equation}
where $\omega = |\kb|$ and $V$ is a quantization volume with periodic boundary conditions, which fixes the summation lattice for $\kb$. 

Let the Fourier transforms of the sampling functions be defined by
\begin{equation}\label{eq:Fourier}
\hat{f}(\omega) =\int_{-\infty}^\infty dt \, {\rm e}^{-i\omega t} f(t)
\end{equation}
and
\begin{equation}
\hat{g}(\kb) = \int d^3 x \,  {\rm e}^{i \kb\cdot \xb}\, g(\xb) \,.
\end{equation} 
Equations \eqref{eq:fnorm} and \eqref{eq:gnorm} imply that $\hat{f}(0) = \hat{g}(0) = 1$. 
Here we assume that the sampling functions, and hence their Fourier transforms, are even, real functions.
The matrices  $A_{i j}$ and $B_{i j}$ which appear in $T$ and hence in the expressions for its moments, may be expressed in terms
of $ \hat{f}$ and $ \hat{g}$. For the case of ${\cal T} = :\dot{\varphi}^2:$, we have
\begin{equation}
A_{j \ell} = \frac{\sqrt{\omega_j \omega_\ell}}{V} \, \hat{f}(\omega_j -\omega_\ell )\, \hat{g}(\kb_j -\kb_\ell ) 
\label{eq:A}
\end{equation}
and 
\begin{equation}
B_{j \ell} = \frac{\sqrt{\omega_j \omega_\ell}}{2 V} \, \hat{f}(\omega_j +\omega_\ell )\, \hat{g}(\kb_j +\kb_\ell ) \,,
\label{eq:B}
\end{equation}
both of which are real and symmetric.

We can now understand why time averaging is essential in four spacetime dimensions. The time average contributes a factor of $\hat{f}^2(\omega_j +\omega_\ell )$ to $\mu_2$ which renders the sum over all modes in Eq.~\eqref{eq:mu2} finite. If we had averaged
only in space, then $\mu_2$ would just contain a factor of $\hat{g}^2(\kb_j +\kb_\ell )$, and receive a divergent contribution from the region
where $\kb_j = -\kb_\ell$, that is, from modes with antiparallel wavevectors.   

In Ref.~\cite{FF2015}, it was argued that there is a dominant contribution to $\mu_n$,  which is
\begin{equation}
M_n = 4 \sum_{j_1 \cdots j_n} B_{j_1 j_2}\, A_{j_2 j_3}\, A_{j_3 j_4} 
\cdots  A_{j_{n-1} j_n}\,B_{j_n j_1}\,,
\label{eq:dom}
\end{equation}
This contribution contains the maximum number of factors of $A_{j \ell} $, which tend to be larger that the corresponding $B_{j \ell} $,
because of the minus sign in the  $\hat{f}(\omega_j -\omega_\ell )$ factor, which allows it to be larger on average than the  
$ \hat{f}(\omega_j +\omega_\ell )$ factor in $B_{j \ell} $. We will assume $M_n$ continues to be the dominant contribution when
spatial averaging is included. If $\hat{f}$ and $\hat{g}$ are non-negative, all of the omitted terms are non-negative,  so  $M_n$ is always 
a lower bound on the exact moment. The construction of  non-negative $\hat{f}$ and $\hat{g}$ is discussed in Ref.~\cite{FF2015} and
in Sec.~\ref{sec:compact}.

We now give the generalization of the discussion in Sec.~IIIA of Ref.~\cite{FF2015} to the case with spatial and temporal averaging.  
Use Eqs.~(\ref{eq:A}) and (\ref{eq:B}) to write
\begin{eqnarray}
M_n &=& C_n \, \int_0^\infty d^3 k_1 \cdots d^3 k_n \;\omega_1 \cdots \omega_n \;
\hat{f}(\omega_1 +\omega_2 )  \hat{g}(\kb_1+\kb_2 ) \, \hat{f}(\omega_2 - \omega_3 )   \hat{g}(\kb_2-\kb_3 )   \cdots \\ \nonumber
&&\hat{f}(\omega_{n-1} - \omega_n )   \hat{g}(\kb_{n-1}-\kb_n ) \,   \hat{f}(\omega_n + \omega_1 )  \hat{g}(\kb_n +\kb_{1} )\,,
\label{eq:Mn} 
\end{eqnarray} 
where
\begin{equation}
C_n = \frac{1}{(2 \pi)^{3n}}\, ,
\label{eq:Cn}
\end{equation}
and we have taken the $V \rightarrow \infty$ limit. In the case that $n=2m$ is even, we can write the above expression as
\begin{equation}
M_{2m} = C_{2m} \,  \int d^3 k d^3 q \; k\, q \; 
[G_{m-1}(\kb, \qb)]^2 \,, 
\label{eq:M2m}
\end{equation} 
where $k = |\kb|$,    $q = |\qb|$,  and we define
\begin{equation}
G_{m-1}(\kb_1, \kb_{m+1}) =  \int d^3 k_2 \cdots d^3 k_m \;
\omega_2 \cdots \omega_m\:
\hat{f}(\omega_1 +\omega_2 )  \hat{g}(\kb_1+\kb_2 ) \, \hat{f}(\omega_2 - \omega_3 )  \hat{g}(\kb_2-\kb_3 )  \cdots
\hat{f}(\omega_m - \omega_{m+1} )   \hat{g}(\kb_m- \kb_{m+1} ) \,.
\label{eq:Gdef}
\end{equation}
These functions satisfy a recurrence relation
\begin{equation}
G_{m+1}(\kb,\qb) = \int d^3\ell \,\ell \,\hat{f}(q-\ell)\hat{g}(\qb- \lb)\, G_{m}(\kb,\lb)
\label{eq:Gm}
\end{equation} 
for $m\ge 0$, where 
\begin{equation}
G_0(\kb,\qb) = \hat{f}(q+k)\hat{g}(\qb+\kb) \,.
\label{eq:G0}
\end{equation}

\subsection{Compactly supported averaging functions}
\label{sec:compact}

In this paper, we assume that both $f(t)$ and  $g(\xb)$ are functions with compact support, and hence describe measurements made in both  a finite time interval and a finite spatial region. This implies that their Fourier transforms, $\hat{f}(\omega)$ and $\hat{g}(\kb)$, decay more slowly than exponentially for large values of their arguments. Starting with $f$, we assume that its support has characteristic width $\tau$ (in a specific example given below, this will be the length of the support), and that its Fourier transform behaves asymptotically as
\begin{equation}\label{eq:fasymptgen}
\hat{f}(\omega)\sim C_f {\rm e}^{-|\omega\tau|^\alpha}, \qquad |\omega|\to\infty
\end{equation}
for some constants $0<\alpha<1$ and $C_f>0$, the latter of which is fixed by the requirement that $f$ has unit integral, i.e., $\hat{f}(0)=1$.
It is further assumed that $f$ is even and nonnegative, and that the same is true of $\hat{f}$. A class of functions with these properties was constructed and discussed in detail in Sect. II of Ref.~\cite{FF2015}.

Turning to $g$, we require similar properties and, additionally, spherical symmetry. Functions of this type may be constructed as follows. Start with a nonnegative even and smooth function of compact support, $h$, with support of characteristic width $\ell$ (in an example below, this will be \emph{half} the width of the support)
and Fourier transform obeying
\begin{equation}\label{eq:hasymptgen}
\hat{h}(\omega) \sim C_h  {\rm e}^{-\eta|\omega \ell|^\lambda}, \qquad |\omega|\to\infty
\end{equation}
for some constants $\eta>0$, $0<\lambda<1$ and $C_h>0$. We also assume that $\hat{h}(\omega)$ has a maximum at $\omega=0$ and is monotone decreasing on the positive half-line, so that $\hat{h}'(\omega) \leq 0$ and $\hat{h}''(0) < 0$.
Setting 
\begin{equation}
g(\xb) = \frac{h(|\xb|/\ell)}{2\pi\, \ell^3  |\hat{h}''(0)|},
\end{equation}
we then have
\begin{equation}
\label{eq:ghat}
\hat{g}(\kb) = -\frac{2 }{k \, \ell^3 |\hat{h}''(0)|}\, \frac{d}{dk} \int_0^\infty dr\, \cos (k r) \, h(r/\ell) = \frac{\hat{h}'(k \ell)}{k \ell\, \hat{h}''(0)}  \,.
\end{equation} 
Using L'H\^opital's rule and the fact that $\hat{h}'(0)=0$ it is easily seen that $\hat{g}(0)=1$, so $g$ has unit integral over $3$-space. Note also that $\hat{g}(\kb) \geq 0$ for all $\kb$.
Furthermore, we may deduce
\begin{equation}
\hat{g}(\kb) \sim C_g  \frac{{\rm e}^{-\epsilon k^\lambda}}{k^{2-\lambda}}\qquad \text{as}~ \kb\to\infty\,,
\label{eq:gasymptgen}
\end{equation}
where 
\begin{equation}
\epsilon= \eta s^\lambda\,,\qquad C_g = \frac{\lambda\epsilon C_h}{|\hat{h}''(0)|}\,.
\end{equation} 
Here we define $s = \ell/\tau$ as the ratio of the spatial and temporal sampling widths. We will henceforth adopt units of time in which $\tau=1$, so $s = \ell$, unless otherwise noted. 
 In this situation, $\epsilon^{1/\lambda}$ measures the ratio of spatial and temporal sampling scales.

A specific example for the case $\alpha=\lambda=\tfrac{1}{2}$ may be based on results in~\cite{FF2015}, where a nonnegative smooth and even function $L$ was constructed, with support $[-1,1]$, unit integral, and nonnegative Fourier transform obeying
\begin{equation}
\hat{L}(\omega)\sim C_L {\rm e}^{-\sqrt{2|\omega|}}\qquad \text{as}~ |\omega|\to\infty\,,
\end{equation}
where the numerical value of $C_L=2.9324$ to $5$ significant figures. See in particular Figs.~4 \& 5 of Ref.~\cite{FF2015}.
Setting 
\begin{equation}
f(t) = \frac{2}{\tau} L(2t/\tau)\,, \qquad h(r)=L(r/s)\,,
\end{equation}
then $f$ has support $[-\tau/2,\tau/2]$, while $g$ is supported in a ball of radius $s$. 
Noting that $\hat{f}(\omega)=\hat{L}(\omega\tau/2)$ and $\hat{h}(\omega)=s \hat{L}(\omega s)$, the transforms of $f$ and $g$ have asymptotic behavior 
\begin{equation}
\hat{f}(\omega)\sim C_f {\rm e}^{-\sqrt{|\omega\tau|}}\qquad \text{as}~ |\omega|\to\infty\,,
\end{equation}
where $C_f=C_L$, and
\begin{equation}
\hat{g}(\kb)\sim C_g \frac{{\rm e}^{-\epsilon\sqrt{k}}}{k^{3/2}}\qquad \text{as}~ \kb\to\infty\,,
\end{equation}
where $\epsilon=\sqrt{2s}$ and $C_g$ has numerical value
\begin{equation}
C_g = \frac{27.18}{s^{3/2}}
\label{eq:Cg}
\end{equation}
The construction of some specific approximate forms for $\hat{f}(\omega)$ and $\hat{g}(k)$  is described in more detail in Appendix~A.

\section{Analysis of the iteration procedure}
\label{sec:iteration}

\subsection{Heuristic treatment}\label{sec:heuristic}

Any smooth compactly supported function has a Fourier transform that decays faster than any inverse power. Therefore the integrals in Eq.~\eqref{eq:Gm} are 
dominated by contributions from certain regions of the integration domain. Proceeding somewhat heuristically for the moment, the factor of $\hat{f}$ restricts the 
effective integration region to a shell of typical radius $\sim q$ and thickness $\rho_{\hat{f}} \propto1/\tau$, while the factor of $\hat{g}$ 
restricts the effective integration region to a ball centered at $\qb$ and of radius $\rho_{\hat{g}}\propto1/s$. Overall, the integration 
will be dominated by contributions arising from the intersection of the ball and shell, as illustrated by Fig.~\ref{fig:ball-and-shell}. 

\begin{figure}
	\begin{center} 
		\begin{tikzpicture}[scale=0.5]
		\draw[fill=orange] (0,0) circle (5);
		\draw[fill=white] (0,0) circle (4);
		\draw[pattern=north west lines, pattern color = green] (30:4.5) circle (2);
		\draw[line width=1pt,->] (0,0) -- (30:4.5) node[right]{$\qb$};
		\node[below] at (0,0) {$\boldsymbol{0}$};
		\end{tikzpicture}\hfil
		\begin{tikzpicture}[scale=0.5]
		\draw[fill=orange] (0,0) circle (2.5);
		\draw[fill=white] (0,0) circle (1.5);
		\draw[pattern=north west lines, pattern color = green] (0:2) circle (5);
		\draw[line width=1pt,->] (0,0) -- (0:2) node[right]{$\qb$};
		\node[below] at (0,0) {$\boldsymbol{0}$};
		\end{tikzpicture}
	\end{center}
	\caption{The ball and shell geometry, indicating the regime where $q$ is larger than the ball radius, in which the effects of spatial averaging are seen (left-hand figure), and the 
	regime where $q$ is smaller than the ball radius and spatial averaging is less significant (right-hand figure).} 
	\label{fig:ball-and-shell}
\end{figure}
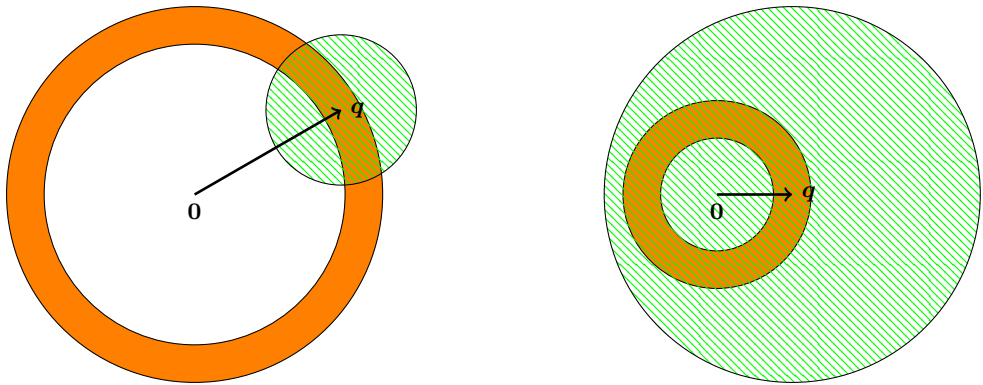

If $q$ is small in relation to the ball radius $\rho_{\hat{g}}$, the shell is contained within the ball so the integration therefore extends over the whole of the shell, 
which has a volume $\sim q^2 \rho_{\hat{f}}$. Therefore one expects, roughly, that 
\begin{equation}
\label{eq:smallq}
G_{m+1}(\kb,\qb)\sim C q^3 \,G_{m}(\kb,\qb)
\end{equation}
for such $\qb$ and a constant $C$. This is the growth rate expected in the worldline limit treated in Ref.~\cite{FF2015}, and corresponds to the factor of 
$\Omega^p$ in Eqs.~(77) and (78) 
of that paper, as we are currently dealing with the case $p=3$. On the other hand, as $q$ becomes large in relation to the radius of the ball
determined by $\hat{g}$, the effective integration region volume tends to a constant $\sim (\rho_{\hat{g}})^2\rho_{\hat{f}}$, where $\rho_{\hat{g}}$ is the effective 
support radius of $\hat{g}$ 
and similarly for $\rho_{\hat{f}}$. Therefore, for large $q$, we expect
\begin{equation}\label{eq:largeq}
G_{m+1}(\kb,\qb)\sim C' q \,G_{m}(\kb,\qb)\,,
\end{equation}
for another constant $C'$. The consequence of this is that low moments (which are largely fixed by the small $q$ regime) will behave
like those of the worldline averaged quantities, whereas higher moments grow rather less rapidly. 
The distinction between low and high moments is determined by the ratio $\rho_{\hat{g}}/\rho_{\hat{f}} \approx \tau/s$: 
the {\em smaller} the scale of {\em spatial} averaging relative to temporal averaging, 
i.e., the {\em larger} the ratio $\rho_{\hat{g}}/\rho_{\hat{f}}$ of {\em momentum space} averaging scales, the larger $q$ must be to detect the 
effect of spatial averaging and therefore the higher 
the threshold beyond which the moments $M_n$ are affected by the spatial averaging.  This fits in with some basic intuition: on one hand, if one shrinks the spatial averaging to a 
$\delta$-function, one ought to obtain the worldline results, consistent with Eq.~\eqref{eq:smallq}; on the other, one would expect that broadening the spatial averaging 
should suppress the 
effects of short-wavelength modes relative to the worldline case and therefore diminish the probability of large fluctuations. 
These expectations are in agreement with the exact results found for CFTs in Sec.~\ref{sec:CFT}.
For the energy density, in fact, if averaging extends uniformly across a full Cauchy surface, one obtains a multiple of the Hamiltonian and all fluctuations vanish 
because the vacuum is  an eigenstate of the Hamiltonian. 
Note, however, that the vacuum is not an eigenstate of the operators formed by integrating stress tensor components other than the energy density over all space.
Nonetheless, we will find that spatial averaging of these stress tensor components also reduces the probability of large vacuum fluctuations.

In the rest of this section we investigate these heuristic ideas more quantitatively by both numerical and analytic means.  

\subsection{The first iteration}
\label{sec:first_iteration}

To start, we consider in more detail how to approximate the first iterate $G_1(\kb,\qb)$, given by
\begin{equation}
G_1(\kb,\qb) = \int d^3\lb\,\ell\hat{f}(q-\ell)\hat{g}(\qb-\lb) \hat{f}(k+\ell)\hat{g}(\kb+\lb)\,,
\end{equation}
in the regime where $\qb$ and $\kb$ both tend to infinity though not necessarily at the same rate.
Each of the Fourier transforms in the integrand decays rapidly as the magnitude of its argument increases. 
Therefore the dominant contributions to the integral are expected to arise from regions where $\lb\approx \qb$ or $\lb\approx -\kb$. Unless $\kb\approx -\qb$, a case that we defer for the moment, these two regions are well-separated as $\qb,\kb\to\infty$ and their contributions may be analysed separately.

Consider first the contribution from $\lb\approx\qb$. In this region, $\hat{f}(k+\ell)\hat{g}(\kb+\lb)\approx \hat{f}(k+q)\hat{g}(\kb+\qb) = G_0(\kb,\qb)$, and therefore the contribution to $G_1$ is expected to be approximately
\begin{equation} \label{eq:G1A}
q  I(q) G_0(\kb,\qb)
\end{equation}
where the function $I(q)$ is defined as
\begin{equation}
I(q) = \int d^3 \lb\, \hat{f}(q-\ell) \hat{g}(\qb-\lb)\,,
\label{eq:Iq}
\end{equation}
and will be called the \emph{iteration coefficient}; note that it depends only on the magnitude $q$ of $\qb$ due to spherical symmetry of $g$. 
The iteration coefficient will be studied in more detail below; in particular, it has a finite, non-zero limit as $q\to\infty$. 

On the other hand, in the region where  $\lb\approx -\kb$ we may approximate $\hat{g}(\qb-\lb)\hat{f}(k+\ell)\approx \hat{g}(\qb+\kb) \hat{f}(2k)$, maintaining the assumption that $\kb\not\approx -\qb$. The contribution is then approximately 
\begin{equation}\label{eq:G1preB}
k \hat{g}(\qb+\kb) \hat{f}(2k) \int d^3\lb\, \hat{f}(q-\ell) \hat{g}(\kb+\lb)\,.
\end{equation} 
Under the additional assumption that $q\gg k$ the $\hat{f}$ factor may be taken outside the integral, 
using $\hat{f}(q-\ell)\approx \hat{f}(q-k) \approx \hat{f}(q+k)$, giving
an approximate contribution 
\begin{equation}\label{eq:G1B}
k \hat{g}(\qb+\kb) \hat{f}(q+k) \hat{f}(2k) \int d^3\lb\, \hat{g}(\kb+\lb) = 
(2\pi)^3 g(0) k\hat{f}(2k) G_0(\kb,\qb)
\end{equation}
to $G_1$. Owing to the rapid decay of $\hat{f}(2k)$, this contribution is subdominant relative to that of Eq.~\eqref{eq:G1A} and we deduce that
\begin{equation}
G_1(\kb,\qb)\approx q  I(q) G_0(\kb,\qb)
\end{equation}
as $\qb,\kb\to\infty$ with $q\gg k$. 
Alternatively, suppose that $\qb$ and $\kb$ have comparable magnitudes. Provided that $\kb\not\approx -\qb$, 
we may then approximate Eq.~\eqref{eq:G1preB} using $\hat{f}(2k)\approx \hat{f}(k+q)$, and replacing $q$ by $k$ under the integral. 
Then Eq.~\eqref{eq:G1preB} contributes approximately 
$k I(k) G_0(\kb,\qb)$ to $G_1(\kb,\qb)$. Combining with Eq.~\eqref{eq:G1A} we have in total
\begin{equation}
G_1(\kb,\qb)\approx  [qI(q)+kI(k)] G_0(\kb,\qb)
\end{equation}
as $\qb,\kb\to\infty$ with $\kb\not\approx -\qb$. In particular, 
\begin{equation}
G_1(\qb,\qb)\approx 2 q I(\qb)G_0(\qb,\qb)  \sim 2 q \, I(\infty)\, G_0(\qb,\qb) 
\label{eq:G1parallel}
\end{equation}
as $\qb\to\infty$. 

If $\kb\approx -\qb$ the two contributing regions overlap and should not be analysed separately. Instead, we expect that
\begin{equation}
G_1(-\qb,\qb) \approx q \hat{f}(2q)  \int d^3 \lb\, \hat{f}(q-\ell) \hat{g}(\qb-\lb)^2  < q \hat{f}(2q)  \int d^3 \lb\, \hat{f}(q-\ell) \hat{g}(\qb-\lb)
=   q \,\hat{f}(2q)\, I(\qb)  \,,
\label{eq:G1antiparallel}
\end{equation}
where the inequality arises because $0 \leq \hat{g} \leq 1$.
 
The ability to  pull factors such as $ \hat{f}(2q)$ out of the integral arises because these functions become flat for large arguments, as was noted above Eq.~(77) in
\cite{FF2015}. More precisely,   $\hat{f}'(\omega)/\hat{f}(\omega) \rightarrow 0$ as $\omega \rightarrow \infty$,  so $\hat{f}'=o(\hat{f})$. In addition,
the function $\hat{h}$ defined in Appendix A satisfies $|\hat{h'}|/\hat{h} \alt 0.33$, and is hence relatively flat for all values of its arguement.

\subsection{The iteration coefficient}
\subsubsection{Form for large $q$}

Our basic hypothesis is that under the iteration Eq.~\eqref{eq:Gm},
\begin{equation}
G_{m+1}(\kb,\qb) \approx  q I(\qb) G_{m}(\kb,\qb)
\end{equation}
for $q\gg k$,
where the iteration coefficient, $I(\qb)$, was defined in Eq.~\eqref{eq:Iq}.
Changing variables to $\mb=\qb-\lb$, 
\begin{equation}
I(q)=\int d^3\mb \fhat(q-\|\qb-\mb\|) \ghat(\mb)
\end{equation}
Our aim is to show that $I(q)\rightarrow I(\infty)$ as $q \rightarrow \infty$, where
\begin{equation}
I(\infty) =  \int d^3\mb \fhat(\qbh\cdot \mb ) \ghat(\mb) \, ,
\label{eq:Iinf1}
\end{equation}
and $\qbh=\qb/q$ is a unit vector along $\qb$.

To prove this, note that for each fixed $\mb$, one has 
\begin{equation}
q-\|\qb-\mb\|= q (1- (1- 2\qbh\cdot \mb/ q + m^2/q^2)^{1/2}) \rightarrow \qbh\cdot \mb
\end{equation}
as $\|\qb\|\to\infty$. Therefore the integrand approaches the required form pointwise.
Noting also that $\fhat(\omega)\le \fhat(0)$ for all $\omega$, and that $\fhat(0)\ghat(\mb)$ is
integrable, the required result follows by the dominated convergence theorem. 
We call $I(\infty)$  the asymptotic iteration coefficient, and identify it with the constant $C'$ which appeared in Eq.~\eqref{eq:largeq}.

\subsubsection{A coordinate space form of $I(\infty)$}

We may write Eq.~\eqref{eq:Iinf1} as
\begin{equation}
I(\infty) = 2 \pi \int_0^\infty dm\, m^2 \, \hat{g}(m) \int_{-1}^1 dc \, \hat{f}(mc) = 2 \pi \int_0^\infty dm\, m \, \hat{g}(m) \int_{-m}^m d\xi \, \hat{f}(\xi)\,,
\label{eq:Iinf2}
\end{equation}  
where $c$ is the cosine of the angle between $\mb$ and  $\qb$, and we let $\xi = m\, c$. Next we use Eq.~\eqref{eq:Fourier} and perform the $\xi$-integration
to write
\begin{equation}
I(\infty) = 2 \pi  i\,\int_{-\infty}^\infty dt \, \frac{f(t)}{t} \int_0^\infty dm\, m \, \hat{g}(m)\, ({\rm e}^{-imt} - {\rm e}^{imt})\,.
\end{equation} 
Next use Eq.~\eqref{eq:ghat} and the fact that $\hat{h}'(m s)$ is an odd function to write 
\begin{equation}
I(\infty) = \frac{2 \pi  i}{s \,\hat{h}''(0)} \, \int_{-\infty}^\infty dt \, \frac{f(t)}{t}  \int_{-\infty}^\infty dm  \,\hat{h}'(m s) \, {\rm e}^{-imt}
=  -\frac{2 \pi}{s^2 \,\hat{h}''(0)} \   \int_{-\infty}^\infty dt \, f(t)\, \int_{-\infty}^\infty dm  \,\hat{h}(m s) \, {\rm e}^{-imt}\,.
\end{equation} 
In the second step above, an integration by parts was performed using $\hat{h}(m s) \rightarrow 0$ as $m \rightarrow \pm \infty$.
Finally, we recognize that the $m$-integration is an inverse Fourier transform yielding $2 \pi h(-t/s) = 2 \pi h(t/s)$ to obtain
\begin{equation}
 I(\infty) = -\frac{4 \pi^2}{s^3 \,\hat{h}''(0)} \   \int_{-\infty}^\infty dt \, f(t)\, h(t/s) \,.
 \label{eq:Iinf-coord}
\end{equation} 
We may use Eq.~\eqref{eq:Fourier} to write 
\begin{equation}
\hat{h}''(0) = -  \int_{-\infty}^\infty dt \, t^2 \, h(t) = -2 \int_{0}^\infty dt \, t^2 \, h(t) \,,
\end{equation} 
which allows $I(\infty)$ to be calculated directly from the coordinate space sampling functions, $f(t)$ and $h(t)$.

Recall that $f(t)$ has a characteristic width $\tau =1$, and $h(t/s)$ has width $s$. It is of interest to consider the limits in which
one of these widths is large compared to the other. First consider the case of a large spatial sampling region, $s \gg 1$. This causes
$h(t/s)  \approx h(0)$, and we may use $ \int_{-\infty}^\infty dt \, f(t) =1$ to write
\begin{equation}
 I(\infty) \approx -\frac{4 \pi^2 \, h(0)}{s^3 \,\hat{h}''(0)} \, , \qquad s \gg 1\,.
 \label{eq:large-s}
\end{equation} 
In the opposite limit of a small spatial sampling scale, we note that the function $h(t/s)$ forces the integral to get its dominant contribution
from small $t$, so  $f(t)  \approx f(0)$, and now we use $ \int_{-\infty}^\infty dt \, h(t/s) =s$ to find
\begin{equation}
 I(\infty) \approx -\frac{4 \pi^2 \, f(0)}{s^2 \,\hat{h}''(0)} \, , \qquad s \ll 1\,.
 \label{eq:small-s}
\end{equation} 
The powers of $s^{-3}$ and $s^{-2}$ which appear in Eqs.~\eqref{eq:large-s} and \eqref{eq:small-s}, respectively, will be numerically confirmed in
Sec.~\ref{sec:tail-form}.

 \subsection{Test of the iteration procedure}
 
 Here we wish to test numerically a special case of our proposed iteration procedure. Specifically, we expect that
 \begin{equation}
G_{1}(\kb,\qb) \approx q \, I(\infty) \, G_{0}(\kb,\qb)\,,
\end{equation} 
in the limit that $q \gg k$. Define 
\begin{equation}
R = \frac{G_{1}(\kb,\qb)}{q \, I(\infty) \, G_{0}(\kb,\qb)}\,. 
\end{equation} 
We numerically evaluate $G_{1}(\kb,\qb)$ and $G_{0}(\kb,\qb)$, using Eqs.~(\ref{eq:Gm}) and (\ref{eq:G0}), Here we use the approximate forms of
$\hat{f}(\omega)$ and $\hat{g}(k)$ for the case $\alpha = \lambda =1/2$ given in Appendix~A.

  \begin{figure}[htbp]
	\centering
		\includegraphics[scale=0.3]{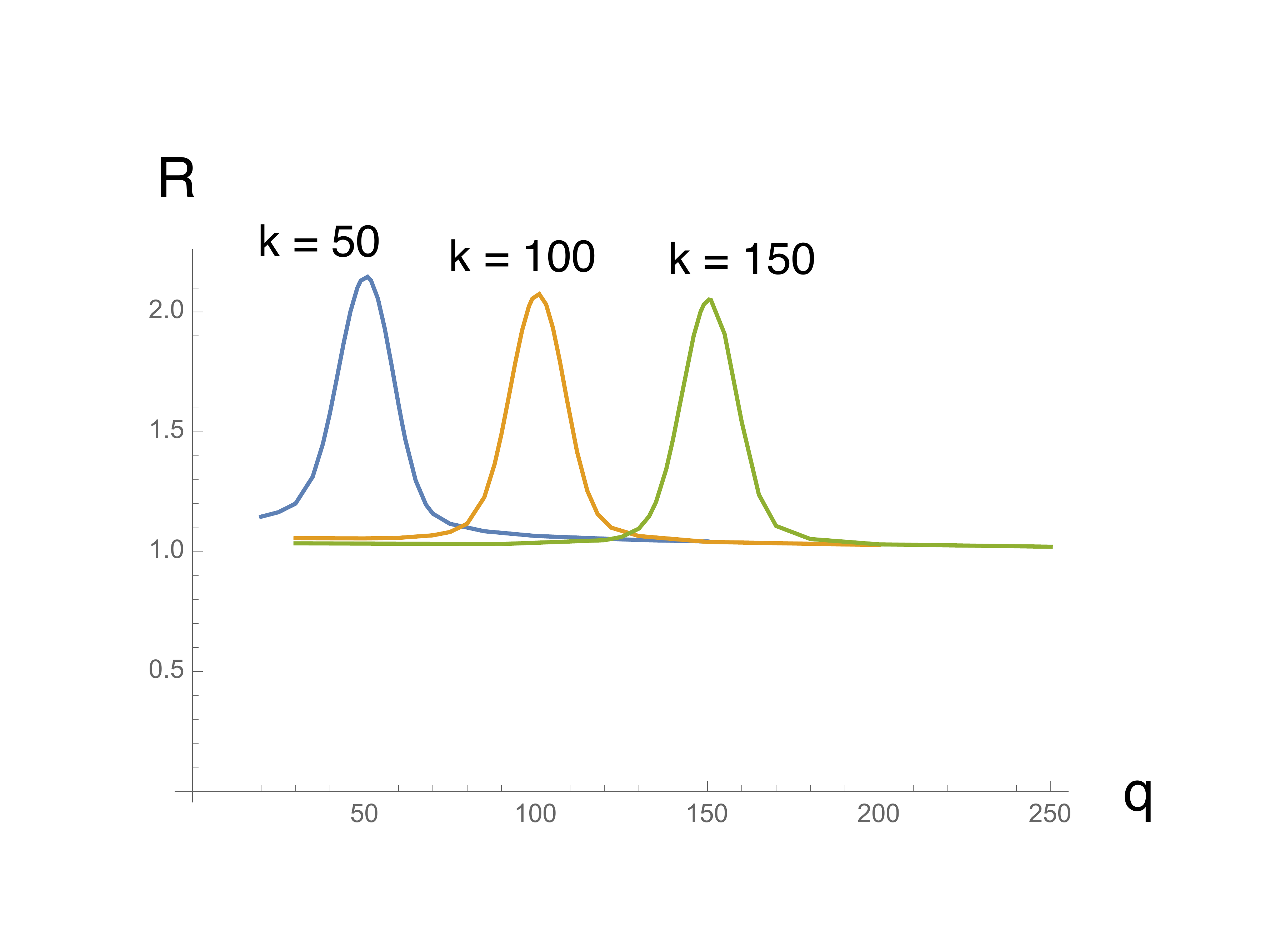}
		\caption{The ratio $R$ of the iteration integral to its expected asymptotic value for large $q$ is plotted as a function of $q$ for three choices of
		$k$ when $\mathbf{q}$ and $\mathbf{k}$ are parallel. 
		 Note that there is a local maximum  when $q \approx k$, but $R \rightarrow 1$ when  $q \gg k$.}
 	\label{fig:qkplot1}
\end{figure}
 \begin{figure}[htbp]
	\centering
		\includegraphics[scale=0.3]{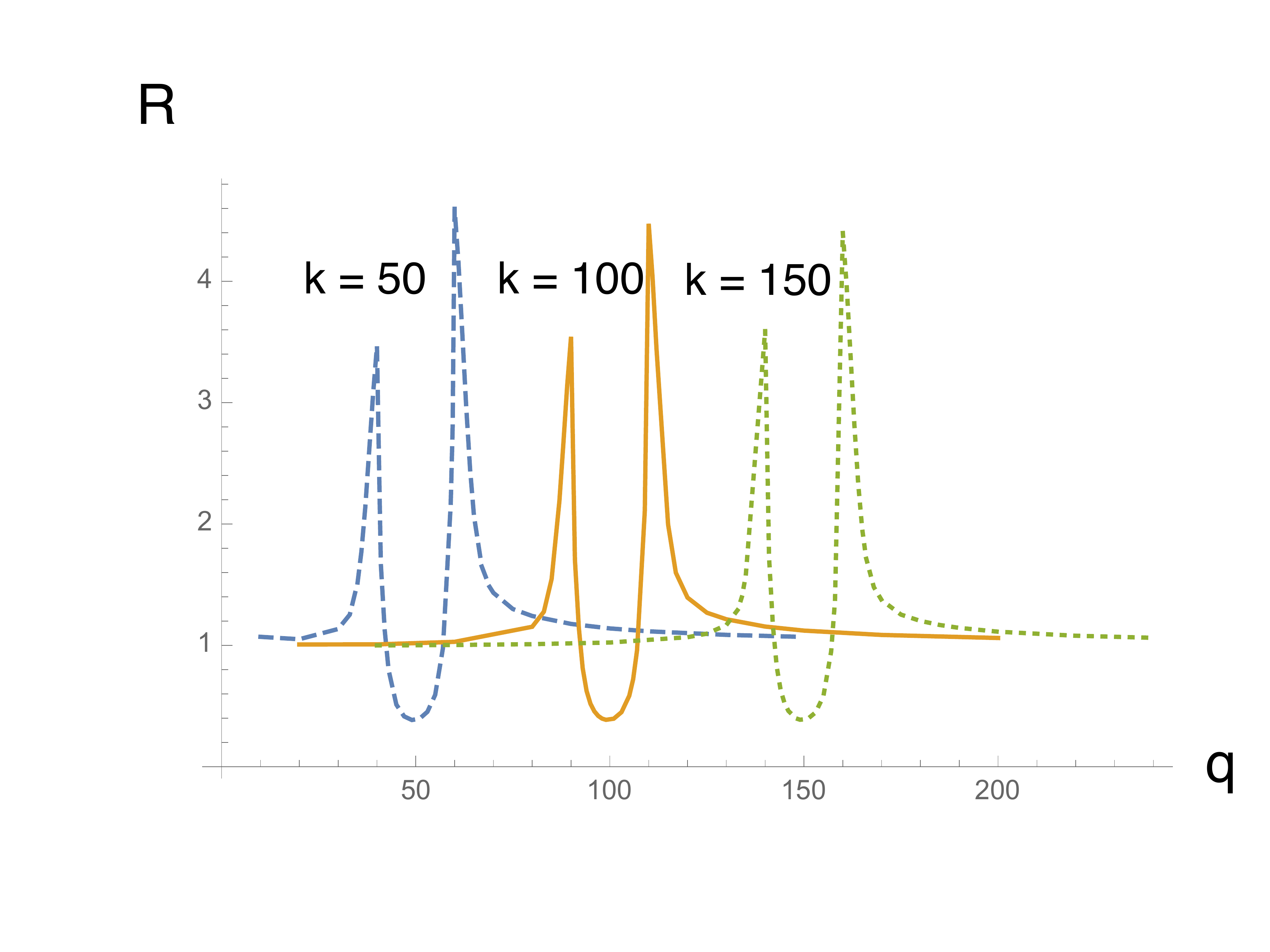}
		\caption{The ratio $R$ as a function of $q$ is repeated for the case that $\mathbf{q}$ and $\mathbf{k}$ are antiparallel. 
		Now there is a local minimum  when $q \approx k$, surrounded by local maxima, but again  $R \rightarrow 1$ when  $q \gg k$.}
 	\label{fig:qkplot2}
\end{figure}

 The ratio $R$  is plotted  in Fig.~\ref{fig:qkplot1} as a function of $q$  for different values of $k$ when the vectors $\mathbf{q}$ and $\mathbf{k}$ are
 parallel, and in  Fig.~\ref{fig:qkplot2} when they are antiparallel.  We see that  $R\approx 1$ for large $q$, which supports our iteration hypothesis. 
We may use the results in Sec.~\ref{sec:first_iteration} to understand some of the other features in Figs.~\ref{fig:qkplot1} and ~\ref{fig:qkplot2}. 
First, there  are maxima in Fig.~\ref{fig:qkplot1} near  $\qb \approx \kb$ where $R \approx 2$. This follows from Eq.~\eqref{eq:G1parallel}, which further 
shows that the height of this ridge is bounded, so $R \rightarrow 2$ when $q \rightarrow \infty$ with $\kb = \qb$. A second feature are the minma in
Fig.~\ref{fig:qkplot2}  near $\qb \approx -\kb$, where $R < 1$. This feature 
follows from  Eq.~\eqref{eq:G1antiparallel}.

\subsection{A growth bound}

Alongside the numerical evidence supporting our iteration procedure, it is useful to have analytic worst-case bounds on the growth of $G_m$. 
We assume that there exist constants $C>0$, $0<\alpha<1$, $\tau>0$, $0<\lambda<1$ and $\epsilon>0$ such that
\begin{equation}\label{eq:assumption}
0\le \hat{f}(\omega)\le C e^{-|\omega\tau|^\alpha},\qquad
0\le \hat{g}(\kb)\le C e^{-\epsilon\|\tau\kb\|^\lambda}
\end{equation}
for all $\omega\in\RR$, $\kb\in\RR^3$. As previously, we  adopt units in which $\tau=1$. 
The parameter $\epsilon^{1/\lambda}$ measures the ratio of spatial and temporal sampling scales.

It is useful to establish some rough bounds on the way in which the functions
$G_m$ can grow with $m$. Because it is no more difficult, we study a slightly
more general problem than the recurrence relation expressed by~\eqref{eq:Gm} and~\eqref{eq:G0}. 

For integer $p\ge 1$, and with fixed test functions $f$ and $g$ whose Fourier transforms satisfy Eq.~\eqref{eq:assumption}, we define an integral operator $\Xi^{(p)}$ by
\begin{equation}
(\Xi^{(p)}G)(\kb,\qb) = \int d^3\lb\,\ell^p \hat{f}(q-\ell)\hat{g}(\qb-\lb) G(\kb,\lb)
\end{equation}
and consider the iteration $G_{m+1}=\Xi^{(p)} G_m$, with $G_0$ as in~\eqref{eq:G0}.
 
Starting from the assumption in Eq.~(\ref{eq:assumption}), our aim is to prove that
\begin{equation}\label{eq:Gm_rough}
|G_m(\kb,\qb)|\le Q_m^{(p)}(q) e^{-(k+q)^\alpha-\epsilon \|\kb+\qb\|^\lambda}
\end{equation}
for all $\kb,\qb\in\RR^3$, where $Q_m^{(p)}$ is a polynomial of degree at most $mp$ with coefficients independent of $\qb$ and $\kb$.

In our situation of interest, $p=1$, so the polynomial factor in $q$ has degree at most $m$, which supports the heuristic expectation given in Eq.~\eqref{eq:largeq}. We will need two 
useful inequalities. The first was proved as Eq.~(B6) in~\cite{FF2015}, and asserts 
\begin{equation}
x^\alpha + y^\alpha\ge (x+y)^\alpha +(1-\alpha) \min\{x,y\}^\alpha
\label{eq:B6}
\end{equation}
which holds for $x,y>0$ and $0<\alpha<1$. Here, we also require an analogous inequality on vector
norms,
\begin{align}
\|\xb\|^\alpha + \|\yb\|^\alpha &\ge (\|\xb\|+\|\yb\|)^\alpha + (1-\alpha)\min\{\|\xb\|,\|\yb\|\}^\alpha \notag\\
&\ge 
\|\xb + \yb\|^\alpha + (1-\alpha)\min\{\|\xb\|,\|\yb\|\}^\alpha
\label{eq:B6-v}
\end{align}
for $\xb,\yb\in\RR^3$, $0<\alpha<1$, where in the first step we apply~\eqref{eq:B6} to $x=\|\xb\|$ and $y=\|\yb\|$ and in the second,  we have applied the ordinary triangle inequality, and the 
fact that $0<\alpha<1$.

The proof of Eq.~\eqref{eq:Gm_rough} is inductive. The statement is true 
by assumption for $m=0$, because it follows from Eq.~\eqref{eq:assumption} and Eq.~\eqref{eq:G0} that
\begin{equation}
|G_0(\kb,\qb)|\le  C^2 e^{-(k+q)^\alpha-\epsilon \|\kb+\qb\|^\lambda}
\end{equation}
for all $\kb,\qb\in\RR^3$. So let us now suppose that \eqref{eq:Gm_rough} holds for some $m\ge 0$. 
We obtain
\begin{equation}
|G_{m+1}(\kb,\qb)|\le C^2 \int d^3\ell\,\ell^p Q_m^{(p)}(\ell)
e^{-|q-\ell|^\alpha-\epsilon \|\qb-\lb\|^\lambda} e^{-(k+\ell)^\alpha-\epsilon \|\kb+\lb\|^\lambda}
\end{equation}

Expanding the degree-$mp$ polynomial $Q^{(p)}_m$, it is clearly sufficient for our inductive argument to show that integrals of the form 
\begin{equation}
L^{(r)}(\kb,\qb):= \int d^3\ell\,\ell^r
e^{-|q-\ell|^\alpha-\epsilon \|\qb-\lb\|^\lambda} e^{-(k+\ell)^\alpha-\epsilon \|\kb+\lb\|^\lambda},
\end{equation}  
with $r\ge p\ge 1$, obey bounds of the form
\begin{equation}\label{eq:Lrough}
L^{(r)}(\kb,\qb)\le P^{(r)}(q) e^{-(k+q)^\alpha-\epsilon \|\kb+\qb\|^\lambda}
\end{equation}
for all $\kb,\qb$, where $P^{(r)}$ is a polynomial of degree $r$ with coefficients independent of $\kb$ and $\qb$, whose leading coefficient is also independent of $r$.

To prove the estimate~\eqref{eq:Lrough}, we apply \eqref{eq:B6-v} to obtain
\begin{align*}
L^{(r)}(\kb,\qb) &\le   e^{-\epsilon \|\kb+\qb\|^\lambda} 
\int d^3\ell\,\ell^r e^{-|q-\ell|^\alpha-(k+\ell)^\alpha}
e^{ -\epsilon(1-\lambda)\min(\|\qb-\lb\|,\|\kb+\lb\|)^\lambda} \,.
\end{align*}
Now split the integral into the regions $\ell<2^{1/r}q$ and $\ell\ge 2^{1/r} q$. 
In the first of these, we can use the fact that $\ell^r<2q^r$ if $r\ge 1$,  
further, we apply \eqref{eq:B6} to find
\begin{equation}
e^{-|q-\ell|^\alpha-(k+\ell)^\alpha} \le e^{-(k+q)^\alpha-(1-\alpha)\min(|q-\ell|,k+\ell)^\alpha} \le e^{-(k+q)^\alpha}.
\end{equation}
Thus the contribution is bounded from above by 
\begin{equation}\label{eq:roughinner}
 2q^r  e^{-(k+q)^\alpha-\epsilon \|\kb+\qb\|^\lambda} 
\int_{\ell<2^{1/r}q} d^3\ell\,
e^{ -\epsilon(1-\lambda)\min(\|\qb-\lb\|,\|\kb+\lb\|)^\lambda}.
\end{equation}
In the second region, we use $e^{-(k+\ell)^\alpha}\le e^{-(k+q)^\alpha}$
to see that the contribution is bounded by 
\begin{equation}\label{eq:roughouter}
S_{r,\alpha} e^{-(k+q)^\alpha-\epsilon \|\kb+\qb\|^\lambda} 
\int_{\ell> 2^{1/r}q} d^3\ell\,
e^{ -\epsilon(1-\lambda)\min(\|\qb-\lb\|,\|\kb+\lb\|)^\lambda},
\end{equation}
where 
\begin{align}
S_{r,\alpha}&:=\sup_{q> 0}\sup_{\ell>2^{1/r}q} \ell^r e^{-(\ell-q)^\alpha}  \notag\\
&= \sup_{q> 0}\sup_{\ell>2^{1/r}q} (1-q/\ell)^{-r} (\ell-q)^r e^{-(\ell-q)^\alpha} \notag\\
&\le (1-2^{-1/r})^{-r} \sup_{x>0} x^{r/\alpha} e^{-x} = (1-2^{-1/r})^{-r} (r/\alpha)^{r/\alpha}e^{-r/\alpha}.
\end{align}
As the upper bound suggests, $S_{r,\alpha}$ will grow rapidly in $r$ for fixed $\alpha$. We may recombine the estimates \eqref{eq:roughinner} and \eqref{eq:roughouter} as
\begin{equation}
L^{(r)} (\kb,\qb)\le  (2q^r+ S_{r,\alpha})  e^{-(k+q)^\alpha-\epsilon \|\kb+\qb\|^\lambda} 
\int d^3\ell\,
e^{ -\epsilon(1-\lambda)\min(\|\qb-\lb\|,\|\kb+\lb\|)^\lambda},
\end{equation}
where we have simply estimated the individual integrals by their extension to all of $\RR^3$. Using the elementary fact 
\begin{equation} \label{eq:expest}
e^{-\min\{A,B\}}\le e^{-A}+e^{-B}
\end{equation}
and the freedom to translate the origin of coordinates, one has
\begin{align*}
\int d^3\ell\,
e^{ -\epsilon(1-\lambda)\min(\|\qb-\lb\|,\|\kb+\lb\|)^\lambda} \le
2 \int d^3\ell\,
e^{ -\epsilon(1-\lambda)\ell^\lambda} &=  8\pi \int_0^\infty d\ell\,\ell^2 e^{-\epsilon(1-\lambda)\ell^\lambda} = 
\frac{8\pi\Gamma(3/\lambda)}{\lambda(\epsilon(1-\lambda))^{3/\lambda}} ,
\end{align*}
which gives overall, 
\begin{equation}
L^{(r)}(\kb,\qb)\le \frac{8\pi\Gamma(3/\lambda)}{\lambda(\epsilon(1-\lambda))^{3/\lambda}} 
 (2q^r+ S_{r,\alpha})  e^{-(k+q)^\alpha-\epsilon \|\kb+\qb\|^\lambda} .
\end{equation}
Accordingly,  $L^{(r)}(\kb,\qb)$ 
is bounded by a polynomial in $q$ (with coefficients independent of $\kb$ and $\qb$, and leading coefficient independent of $r$) multiplied by $e^{-(q+k)^\alpha-\epsilon \|\kb+\qb \|^\lambda}$. This concludes the inductive proof of the bound \eqref{eq:Gm_rough}.

We make no claim that this is the tightest possible upper bound that could be derived. However,
the argument is relatively simple and indicates a worst-case growth rate for the functions $G_m(\kb,\qb)$ that is nonetheless broadly in line with the heuristic discussion of 
Sec.~\ref{sec:heuristic}, in the case $p=1$.

\section{Rate of Growth of the Moments }
\label{sec:growth}

\subsection{Approximate Forms of the Moments}
\label{sec:approx-forms}

Recall that in the iteration procedure for $G_{m}(\kb,\qb)$, using Eq.~(\ref{eq:Gm}), we expect for the initial iterations to each bring out a factor 
proportional to $q^3$, and the later iterations to each bring out a factor of $I(\infty)\, q$. Thus, for $m \gg 1$, we expect the asymptotic form for 
 $G_{m}(\kb,\qb)$,  to be
 \begin{equation}
G_{m}(\kb,\qb) \approx C \, [I(\infty)]^m \, q^{m + \mu}\, G_{0}(\kb,\qb) \, 
\end{equation}  
where $C$ and $\mu$ are constants which correct for the possibility that the first several iterations bring out different constants and powers of $q$
than do the later iterations. If we use this form in Eq.~(\ref{eq:M2m}), we find
\begin{equation}
M_n \approx C_n \, C^2\, [I(\infty)]^{n-2} \, S_{n+2\mu -1}\,,
\label{eq:Mnasy}
\end{equation} 
where
\begin{equation}
S_N = \int d^3\qb \,q^{N} \int d^3\kb \, k\,  \hat{f}^2(q+k)   \hat{g}^2(\qb+ \kb) \,.  
\label{eq:SN}
\end{equation} 
We will estimate this integral for the case that $N \gg 1$. As we expect that the dominant contribution comes from  $q \gg k$, we approximate $|\qb+\kb|\approx q$.  
If we assume  that $\hat{f}$ and $\hat{g}$ may be approximated by their asymptotic forms, Eqs.~\eqref{eq:fasymptgen} and \eqref{eq:gasymptgen},   then we have
\begin{equation}
S_N \approx T_N = 16 \pi^2 C_{fg}^2 \int_0^\infty dq \,q^{N+2} \int_0^\infty dk \, k^3\, {\rm e}^{-2(q+k)^\alpha}\, \frac{{\rm e}^{-2\epsilon q^\lambda}}{q^{4-2\lambda}}\,,
\label{eq:SN-approx}
\end{equation} 
where we have written
\begin{equation}
C_{fg}=C_f C_g\,.
\end{equation}
Next let $k=r -q$ to write
\begin{equation}
T_N = 16 \pi^2 C_{fg}^2  \int_0^\infty dq \,q^{N+2(\lambda-1)} \, {\rm e}^{-2\epsilon q^\lambda}\, \int_q^\infty dr (r-q)^3 \, {\rm e}^{-2 r^\alpha}\,
\end{equation} 
Define a new variable $u$ by $r = q (1+u)^{1/\alpha}$ to write to final integral above as
\begin{eqnarray}
\int_q^\infty dr (r-q)^3 \, {\rm e}^{-2 r^\alpha} &=& \frac{q^4}{\alpha}\, \int_0^\infty du \, (1+u)^{1/\alpha- 1} \,[(1+u)^{1/\alpha} - 1]^3 
\,{\rm e}^{-2 q^\alpha (1+u)} \nonumber \\
&\approx& \frac{q^4}{\alpha^4}\, {\rm e}^{-2 q^\alpha}\, \int_0^\infty du \, u^3 \, {\rm e}^{-2 q^\alpha u} = \frac{3}{8 \alpha^4}\, q^{4(1-\alpha)}\,  {\rm e}^{-2 q^\alpha}\, ,
\end{eqnarray} 
where in the second step we used the fact that the dominant contribution comes from the region where  $u \ll 1$ because 
$r \approx q$ when $q \gg k$. 
Thus we have
\begin{equation}
T_N \approx \frac{6 \pi^2 C_{fg}^2}{\alpha^4}\, \int_0^\infty dq \,q^{N+2(1+\lambda)-4\alpha} \,  {\rm e}^{-2 (q^\alpha +\epsilon q^\lambda)} \,.
\label{eq:TN}
\end{equation} 
For the case $\alpha = \lambda$, this integral may be evaluated explicitly to obtain
\begin{equation}
T_N \approx \frac{6 \pi^2C_{fg}^2}{\alpha^5}\, [2(1+\epsilon)]^{(2\alpha -N-3)/\alpha} \, \Gamma\left(\frac{N+3}{ \alpha}  -2 \right) \,.
\label{eq:TN-special}
\end{equation} 
When $\alpha = 1/2$, this becomes
\begin{equation}
T_N = \frac{192 \pi^2C_{fg}^2 \, \Gamma(2N+4)}{[2(1+\epsilon)]^{2N+4}} \,.
\label{eq:TN-half}
\end{equation} 

\subsection{Contribution from $\qb+ \kb \approx 0$}

The result in Eq.~(\ref{eq:SN-approx}), that $S_N \approx T_N$, relies upon the dominant contribution to $S_N$ coming from regions where 
$q \gg k$. when $N \gg1$. However, it is worth examining more carefully the contribution from the region where $\qb+ \kb \approx 0$,
where the argument of $\hat{g}$ becomes small, in order to show that this contribution is small in relation to $T_N$. 
In this region $k\approx q$ and the contribution to $S_N$ is therefore bounded by
\begin{align}\label{eq:SN1}
S_{N1} &= \int d^3\qb\, q^{N+1} \hat{f}(2q)^2 \int d^3\kb\, \hat{g}(\kb+\qb)^2\notag \\
&= 4\pi \int_0^\infty dq\, q^{N+3} \hat{f}(2q)^2\int d^3\kb\, \hat{g}(\kb)^2 \notag\\
&\lesssim C\int_0^\infty dq\, q^{N+3} e^{-2(2q)^\alpha} \notag\\
&\approx C' 2^{-N(1+1/\alpha)} \Gamma\left(\frac{N+4}{\alpha}\right)  
\end{align}  
for constants $C$ and $C'=C/(16^{1+1/\alpha}\alpha)$, depending on $f$, $g$ and $\alpha$ but not $N$. Here we have changed variables from $\kb$ to $\kb+\qb$ in the second line.
We need this contribution to $S_N$ be small compared to  $T_N$, our estimate for $S_N$, when $N$ is large. Next we will examine several special cases.

\subsubsection{ Case: $\alpha = \lambda = \frac{1}{2}$}
Here we have an explicit formula for $T_N$, given in Eq.~(\ref{eq:TN-half}), while 
\begin{equation}
S_{N1} \lesssim \frac{C'}{2^{3N}} \; \Gamma(2 \,N +8)\,.
\end{equation} 
This is suppressed compared to  $T_N$ by a factor proportional to
\begin{equation}
 \left(\frac{1+ \epsilon}{\sqrt{2}} \right)^{2 N } N^4\,.
\end{equation} 
This factor decreases as $N$ grows provided that $\epsilon \leq \sqrt{2} - 1 \approx 0.414$. Under this condition, in which spatial sampling takes place over modest scales relative to temporal sampling, we expect $T_N$ to be a
good approximation to $S_N$ for large $N$ for  $\alpha = \lambda = \frac{1}{2}$. 

\subsubsection{ Case: $\lambda \leq \alpha/2$}

Here we may use some asymptotic results given in Appendix B. First note that if we let $q = 2^{-1/\alpha} \,r$, then Eq.~(\ref{eq:TN})
becomes
\begin{equation}
T_N = \frac{6 \pi^2 C_{fg}^2}{\alpha^4} \, 2^{-(N+3+2\lambda)/\alpha +4}\, \int_0^\infty dr \, r^{N+2(1+\lambda)-4\alpha} \, e^{-r^\alpha -\epsilon' r^\lambda} \,
\propto 2^{-N/\alpha} \, I_{N+3+2\lambda-4\alpha}(\epsilon'),
\label{eq:TNb}
 \end{equation} 
 where $\epsilon' = 2^{1-\lambda/\alpha} \, \epsilon$, and $I_N$ is defined 
 as
 \begin{equation}
 I_N = \int_0^\infty dq\, q^{N-1} e^{-q^\alpha-\epsilon q^\lambda}\,.
 \label{eq:INtext}
 \end{equation} 
 The asymptotic forms of $I_N$ for large $N$ are given in
 Eq.~(\ref{eq:Lt-half}) when $\lambda < \alpha/2$, and in  Eq.~(\ref{eq:eq-half}) when $\lambda = \alpha/2$. Although there is a discontinuity 
 between these two forms at $\lambda = \alpha/2$ in the form of a factor of $e^{\epsilon^2/8}$, both forms have the same dependence upon
 $N$:
 \begin{equation}
 I_N(\epsilon) \propto \Gamma(N/\alpha)\, e^{-\epsilon\, (N/\alpha-1)^{\lambda/\alpha}} \sim  \Gamma(N/\alpha)\, e^{-\epsilon(N/\alpha)^{\lambda/\alpha}} \,.
 \label{eq:IN-asy}
 \end{equation} 
We may combine this result with Eqs.~(\ref{eq:SN1}) and (\ref{eq:TNb}) to write
\begin{equation}
\frac{S_{N1}}{T_N} \propto \frac{\Gamma(N/\alpha+4/\alpha)}{\Gamma((N+3+2\lambda)/\alpha-4)} \, e^{\epsilon' \,(N/\alpha)^{\lambda/\alpha}} \, 
e^{-N \,\ln2 }  \,.
\label{eq:SN1-TN}
 \end{equation} 
The ratio of gamma functions can at most grow as a power of $N$, and here $\lambda/\alpha \leq 1/2$, so the behavior of the ratio ${S_{N1}}/{T_N}$
is dominated by the $e^{-N \,\ln2 }$  factor, which decays exponentially as $N$ increases, leading to ${S_{N1}}  \ll {T_N}$ for large $N$. 

\subsubsection{ Case: $\alpha/2 \leq \lambda \leq 2 \alpha/3$}

The asymptotic form for $I_N$ in this case is given by  Eq.~(\ref{eq:Lt-two-thirds}), where $\beta = \lambda/\alpha$. Note that the exponential in
the right-hand-side of Eq.~(\ref{eq:Lt-two-thirds}) contains two terms. The first is a negative term proportional to $(N/\alpha)^\beta$, which also
appears in  Eqs.~(\ref{eq:Lt-half}) and (\ref{eq:eq-half}). The second is a positive term to proportional to $(N/\alpha)^{2 \beta-1}$.  However,
$\beta > 2 \beta-1$ in the range of interest here, so the first term dominates the exponential and again leads to the same leading order asymptotic
behavior for $I_N$ as that given in Eq.~(\ref{eq:IN-asy}). Hence, the ratio ${S_{N1}}/{T_N}$ is again given by Eq.~(\ref{eq:SN1-TN}) for large $N$.
In all of these cases, we conclude that $S_{N1}$ is asymptotically small compared to $T_N$, so the region where $\qb+ \kb \approx 0$ does
not give a large contribution to $S_N$.

\subsection{Numerical Tests of  $S_N \rightarrow T_N$}
\label{sec:SN-to-TN}

 We can test the approach of $S_N$ to its limiting form, $T_N$, for large $N$ by numerically evaluating Eqs.~\eqref{eq:SN} and \eqref{eq:TN}. In the special
 case that $\lambda = \alpha =1/2$, $T_N$ is given by Eq~\eqref{eq:TN-half}, and we may use the explicit forms for $\hat{f}$ and $\hat{g}$ constructed
 in Appendix~A to evaluate $S_N$.  In all cases, we may approximate the sampling functions in  Eq.~\eqref{eq:SN} by their asymptotic forms for large arguments 
 if $N$ is large. In this case, we use  Eq.~\eqref{eq:fasymptgen} for  $\hat{f}$. However, we need to modify the form given in Eq.~\eqref{eq:gasymptgen} for  $\hat{g}$
 to avoid a singularity at $\qb+ \kb =0$. For this purpose, we use the cutoff-dependent form
 \begin{equation}
 \hat{g}_C(k,Q_0) = C_g  \frac{{\rm e}^{-\epsilon k^\lambda}}{(k+Q_0)^{2-\lambda}}\,,
 \label{eq:cutoffghat}
 \end{equation}
 and test the dependence of the integral upon the parameter $Q_0$.
 
  The results obtained from both approaches are plotted in Fig.~\ref{fig:SNa} for the case that  $\lambda = \alpha =1/2$,  where $\epsilon = \sqrt{2 s}$,
  and agree reasonably well. The cutoff parameter
  $Q_0$ was varied between values of about $1$ and $10$ without a significant effect. We can see that for smaller values of $s$, $S_N/T_N$ becomes close to one for large
  $N$. For larger values of $s$, $S_N/T_N$ is noticeably larger than one for the range of $N$ considered.
   \begin{figure}[htbp]
	\centering
		\includegraphics[scale=0.3]{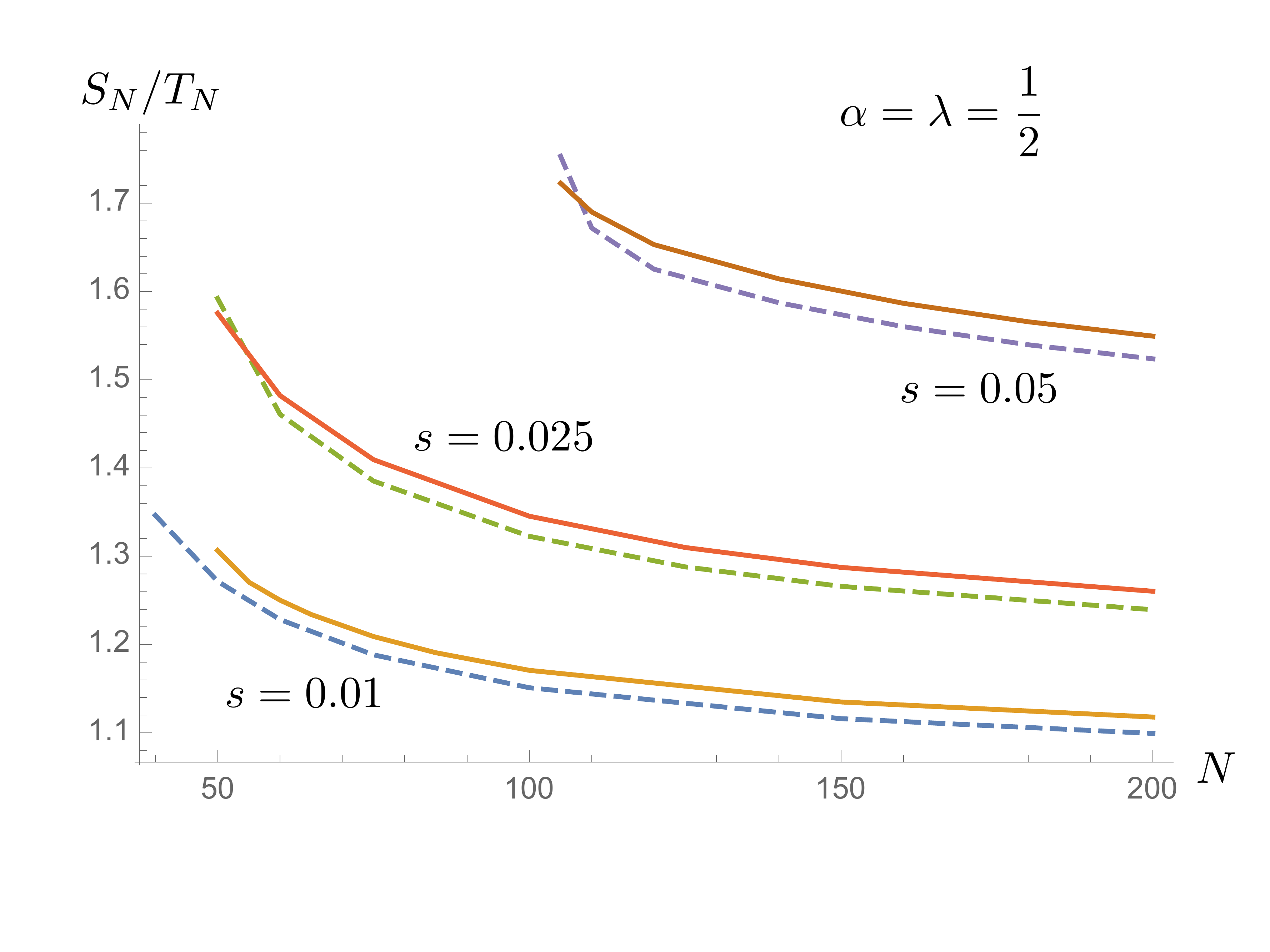}
		\caption{The ratio $S_N/T_N$ is plotted as a function of $N$ for different values of $s$ for the case $\alpha = \lambda = \frac{1}{2}$. The solid lines
		were computed using the forms for $\hat{f}$ and $\hat{g}$ constructed in Appendix~A, and the dashed lines using the asymptotic forms, Eqs.~\eqref{eq:fasymptgen}
		and  \eqref{eq:cutoffghat}.   }
	\label{fig:SNa}
\end{figure}

  Some results for  $\alpha =1/2$, but $\lambda < \alpha$ are plotted in Figs.~\ref{fig:SNb} and ~\ref{fig:SNa}. In this case, Eq.~\eqref{eq:SN} was evaluated using
  Eqs.~\eqref{eq:fasymptgen} and ~\eqref{eq:gasymptgen}. Again, the result seems to be relatively independent of   $Q_0$. Here we appear to find that  
  $S_N \rightarrow T_N$ for $N \gg1$, but that this limit is attained more quickly for smaller values of $\epsilon$ and of $\lambda$. Note that in all cases, 
  we find $S_N > T_N$, 
     \begin{figure}[htbp]
	\centering
		\includegraphics[scale=0.3]{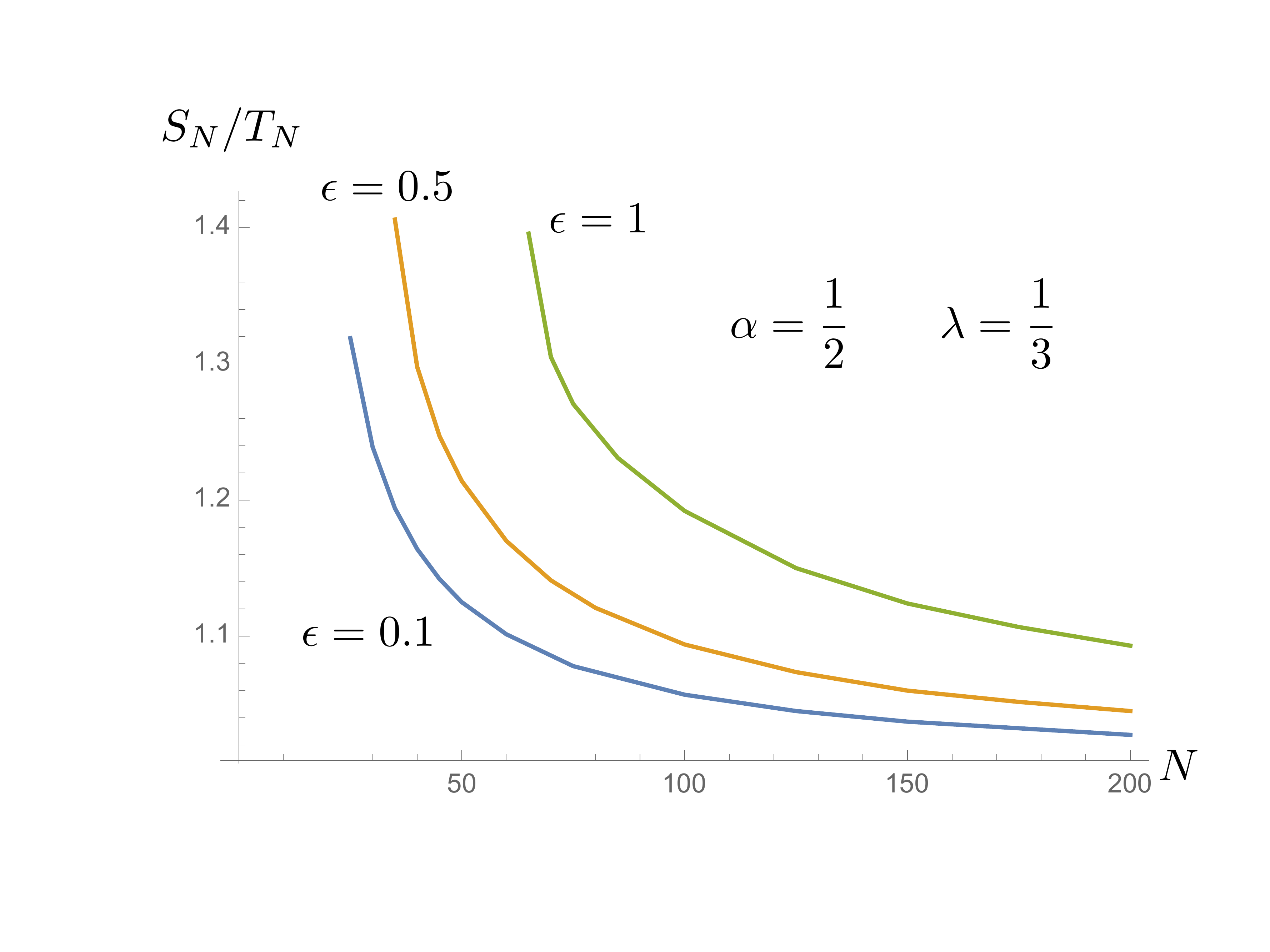}
		\caption{The ratio $S_N/T_N$ is plotted as a function of $N$ for different values of $\epsilon$ for the case $\alpha =  \frac{1}{2}$
		and $\lambda = \frac{1}{3}$.}
	\label{fig:SNb}
\end{figure}	
 
    \begin{figure}[htbp]
	\centering
		\includegraphics[scale=0.3]{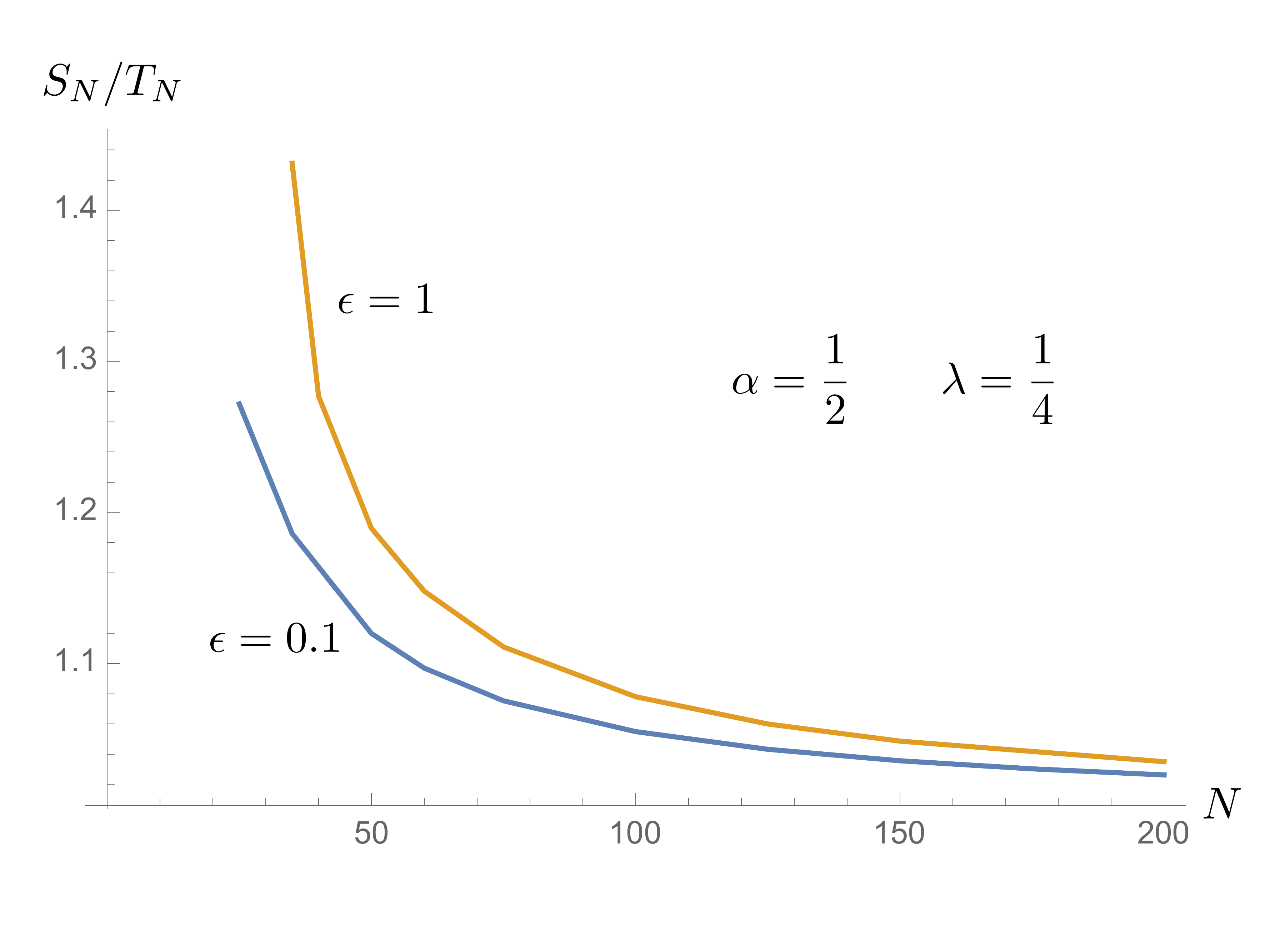}
		\caption{The ratio $S_N/T_N$ is plotted as a function of $N$ for two values of $\epsilon$ for the case $\alpha =  \frac{1}{2}$
		and $\lambda = \frac{1}{4}$. Here this ratio approaches one more quickly and is less dependent upon the value of $\epsilon$,
		as compared with the cases with larger value of  $\lambda$. }
	\label{fig:SNc}
\end{figure}	
  
 In the special case that  $\lambda < \alpha/2 < 1/2$, we are able to give a rigorous proof that $S_N/T_N \rightarrow 1$ as $N \rightarrow \infty$, but the details
 will be omitted here.

\subsection{Asymptotic Behavior of the Moments}

We may now use Eq.~\eqref{eq:Mnasy}  and assume that $S_N \approx T_N$ to write
\begin{equation}
M_n \approx C_n \, C^2\, [I(\infty)]^{n-2} \, T_{n+2\mu -1}\,,
\label{eq:Mnasy2}
\end{equation} 
for $n \gg 1$. If we let $q \rightarrow 2^{-1/\alpha}\, q $ in Eq.~\eqref{eq:TN}, then we have
\begin{equation}
T_N \approx \frac{6 \pi^2C_{fg}^2}{\alpha^4}\,  2^{4-(N+3+2\lambda)/\alpha} \, I_{N+3+2\lambda-4\alpha}(\epsilon')  \,.
\label{eq:TN2}
\end{equation} 
where $\epsilon' = 2^{1-\lambda/\alpha}\, \epsilon$ and $I_N(\epsilon)$ is defined in Eq.~\eqref{eq:IN}. Now we have
\begin{equation}
M_n \approx \frac{6 \pi^2 C_{fg}^2}{\alpha^4}\, \left[\frac{C}{I(\infty)}\right]^2 \, 2^{4-2(1+\mu+\lambda)/\alpha} \, B^n\,  I_{n+2(1+\mu+\lambda)-4\alpha}(\epsilon') \,,
\label{eq:Mnasy3}
\end{equation} 
where we have used Eq.~\eqref{eq:Cn}, and defined
\begin{equation}
B = \frac{I(\infty)}{ 2^{1/\alpha} \, ( 2\,\pi)^{3}} \,.
\label{eq:B-def}
\end{equation} 
As already mentioned, the asymptotic behavior of $I_N$ for large $N$ is discussed for several cases in Appendix~B, where it is found that
$I_N / \Gamma(N/\alpha)$ is bounded as $N\to\infty$. This leads to a factor of $\Gamma\left(\frac{ n+2(1+\mu+\lambda)}{\alpha} -4  \right )$ in $M_n$, 
which reveals that for large $n$, the moments grow no faster than $(n/\alpha)!$ (times a factor growing exponentially in $n$). This is slower than the $(3n/\alpha)!$ 
growth rate found in Ref.~\cite{FF2015} for the case of time 
averaging alone. However, if $\alpha < 1$, it is still faster than $n!$ growth.

 \section{The Tail of the Probability Distribution}
 \label{sec:tail}
 
 \subsection{The form of the tail}
  \label{sec:tail-form}

 Note that Eq.~\eqref{eq:Mnasy3} for $M_n$, the dominant contribution to the $n$-th moment, can be written as
 \begin{equation}
M_n \approx K_0 \, B^n\,  I_{n+2(1+\lambda+\mu)-4\alpha}(\epsilon') =  K_0 \, B^n\,  \int_0^\infty dq \, q^{n+1+2(\lambda+\mu)-4\alpha} e^{-q^\alpha-\epsilon' q^\lambda} \,.
\label{eq:Mnasy4}
\end{equation} 
If we let $x= B\, q$, then this expression becomes
\begin{equation}
M_n \approx K\, \int_0^\infty dx \, x^n \; \left[x^{1+2(\lambda+\mu)-4\alpha} \, e^{-(x/B)^\alpha-\epsilon'\, (x/B)^\lambda} \right]\,,
 \label{eq:MnP}
\end{equation} 
where $K_0$  and $K$  are constants independent of $n$.  
 Recall that the moments of the probability distribution, $P(x)$, are $\mu_n$, where 
 \begin{equation}
M_n \approx \mu_n = \int_{-x_0}^\infty dx \, x^n \, P(x) \approx \int_{0}^\infty dx \, x^n \, P(x)\,.
 \label{eq:MnP2}
\end{equation} 
The last step holds when $n$ is sufficiently large that the the interval  $[-x_0,0]$ makes a negligible contribution to the integral.
Comparison of Eqs.~\eqref{eq:MnP} and \eqref{eq:MnP2} suggests that 
 \begin{equation}
 P(x) \approx K \, x^{1+2(\lambda+\mu)-4\alpha} \, e^{-(x/B)^\alpha-\epsilon'\, (x/B)^\lambda} 
 \label{eq:P}
\end{equation} 
for large $x$. 

This identification is subject to the possible ambiguity that rapidly growing moments may not uniquely determine the 
probability distribution.  However, for a probability distribution which is nonzero on a half line, as is the case here,
the condition that the moments uniquely determine $P(x)$ is the Stieltjes criterion~\cite{Simon}, which requires
\begin{equation}
|\mu_n| \leq C\, D^n\, (2n)!
 \label{eq:Stieltjes }
\end{equation} 
for all $n$ for some choice of constants $C$ and $D$. We found in the previous section that here the moments grow 
no faster than $(n/\alpha)!$, so this criterion is satisfied for $\alpha \geq 1/2$ and hence $P(x)$ is uniquely determined by the moments. If $\alpha < 1/2$, then
we have the same situation as in the worldline case, where the moments might not uniquely determine $P(x)$. Nonetheless, it is possible to gain some information about the tail of the distribution, as discussed in Sec.~VI of Ref.~\cite{FFR2012}. 

The constants $K$ and $\mu$ are not determined by the methods used here, because the transition between the low order and high order iteration regimes, discussed in 
Sec.~\ref{sec:approx-forms}, is not fully understood.
However, the argument of the exponential
in Eq.~\eqref{eq:P} is determined, and governs the primary rate of decay of the tail. If $\lambda < \alpha$, the $(x/B)^\alpha$
term in Eq.~\eqref{eq:P} will eventually dominate the $(x/B)^\lambda$ term, and we will have
 \begin{equation}
 P(x) \propto e^{-(x/B)^\alpha}
 \label{eq:P2}
\end{equation} 
for sufficiently large $x$. In the case that $\lambda = \alpha$, we have the asymptotic form
 \begin{equation}
 P(x) \propto e^{-(1+\epsilon)\,(x/B)^\alpha} \,,
 \label{eq:P3}
\end{equation} 
as $\epsilon' = \epsilon$ in this case. 
Recall that $B$ is determined by Eqs.~\eqref{eq:Iinf2} and \eqref{eq:B-def}. 
In the special case that $\lambda = \alpha= 1/2$,  we may numerically compute $B$  as a function of $s = \ell/\tau$, using the the approximate forms of
$\hat{f}(\omega)$ and $\hat{g}(k)$  given in Appendix~A. The results are illustrated in Figs.~\ref{fig:B1} and \ref{fig:B2}.

 \begin{figure}[htbp]
	\centering
		\includegraphics[scale=0.3]{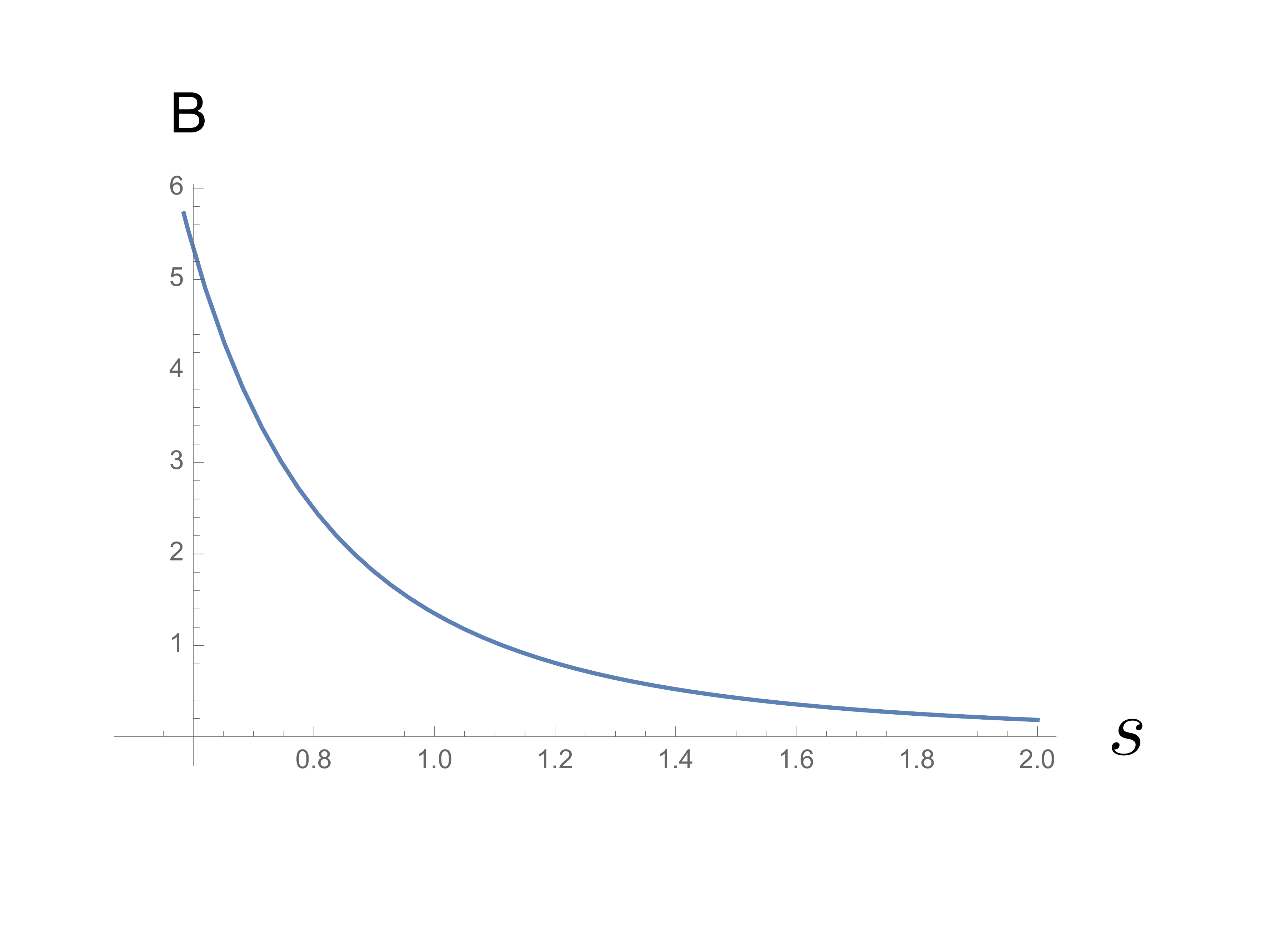}
		\caption{Here the constant $B$, which appears in the asymptotic probability distribution, is plotted as a function of the ratio of the spatial and temporal
		sampling scales, $s = \ell/\tau$ for the case that $\lambda = \alpha= 1/2$.  Note that $B \approx 1$ when $s = 1$, and decreases as $s$ increases. }
			\label{fig:B1}
\end{figure}
 \begin{figure}[htbp]
	\centering
		\includegraphics[scale=0.3]{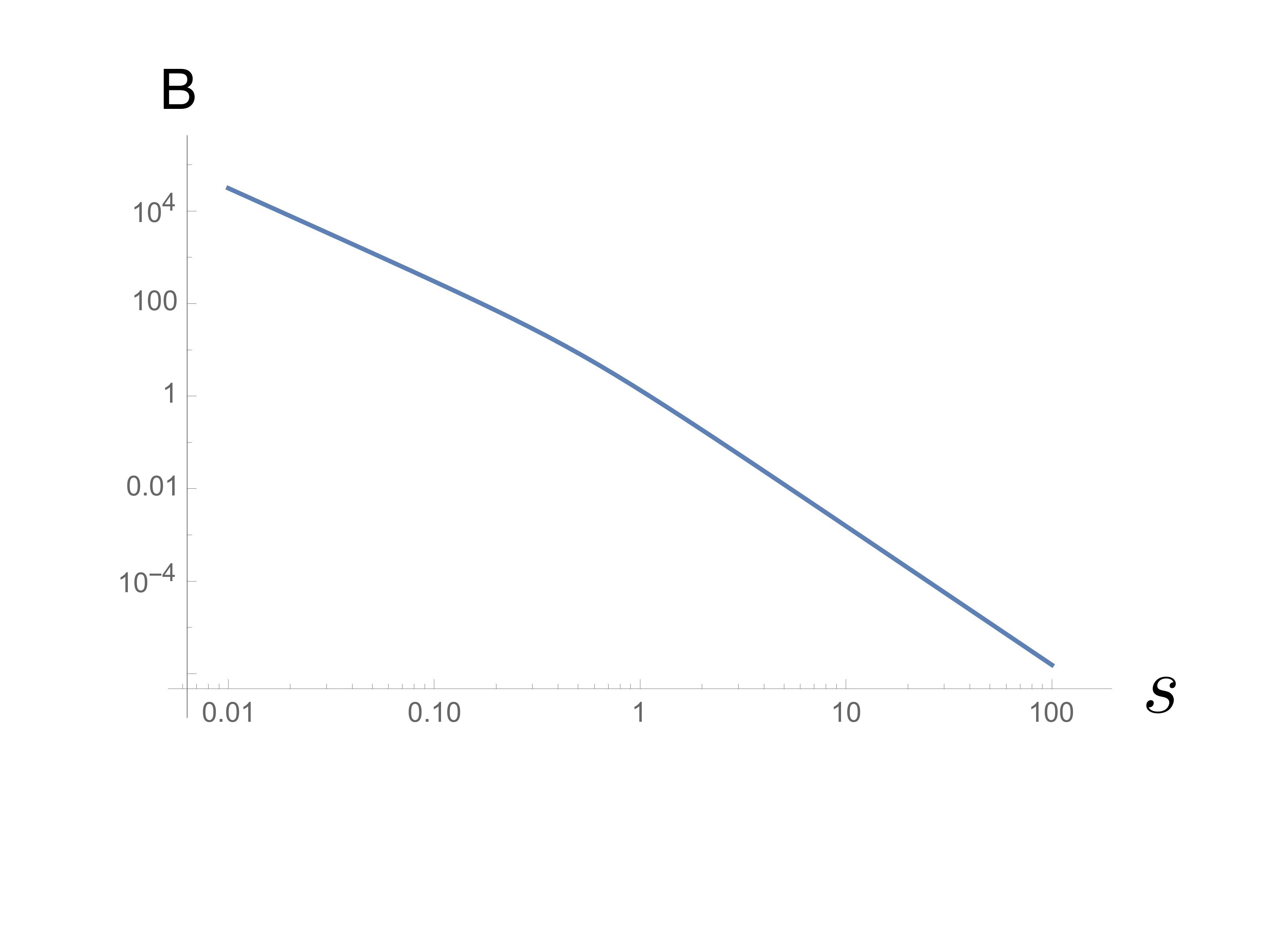}
		\caption{Here $B$ for the case that $\lambda = \alpha= 1/2$ is plotted over a larger range on a log-log plot. Note that $B \propto s^{-2}$ for $s \alt 1$, 
		in accordance with Eq.~\eqref{eq:small-s}, as $B$ decreases from $10^4$ to about $1$ as $s$ increases from $0.01$ to $1$. Furthermore.
		$B \propto s^{-3}$ for $s \agt 1$, in accordance with Eq.~\eqref{eq:large-s}. Here $B$ decreases by about six orders of magnitude as   $s$ increases 
		from $1$ to $100$.}
			\label{fig:B2}
\end{figure}
In all regions, $B$ decreases as $s$ increases. As smaller values of $B$ suppress the probability of a fluctuation with a given dimensionless magnitude $x$, this is 
consistent with the intuition that increasing $\ell$ relative to $\tau$ decreases  the probability of a large fluctuation.

\subsection{The transition from worldline behavior to spacetime averaged behavior}
\label{sec:transition}

Recall that in Ref.~\cite{FF2015}, the averaging along a worldline alone was treated, and the asymptotic form of the probability distribution was found to
be of the form
\begin{equation}
P(x) \sim c_0\, x^b \, {\rm e}^{-a x^c}
\label{eq:asymp-P}
\end{equation} 
with $c = \alpha/3$. In contrast, the asymptotic form of the spacetime averaged distribution, for $\lambda \leq \alpha$, has a similar form, but with $c = \alpha$. 
The effect of the spatial averaging has been to enhance the rate of decrease of the tail of $P(x)$. However, if the spatial sampling scale  $s$ is small compared to
the temporal scale $\tau$, we expect a finite region in $x$ where the worldline form holds approximately. This is the regime depicted in the right part of 
Fig.~\ref{fig:ball-and-shell}, when $q \alt 1/s$ in $\tau =1$ units, and when each iteration produces a factor of $q^3$, as predicted by Eq.~\eqref{eq:smallq}. 
In this regime, the $n$-th moment, given by Eq.~\eqref{eq:M2m}, will contain an integral on $q$ of the form 
\begin{equation}
\int_0^\infty dq \, q^{3n+3}\, \hat{f}^2(q) \approx C_f^2 \int_0^\infty dq \, q^{3n+3}\, {\rm e}^{-2 q^\alpha}\,,
\end{equation} 
where we assume $n \gg 1$ and use Eq,~\eqref{eq:fasymptgen}.
The peak of this integrand, and hence the region which gives the dominant contribution to the integral,  occurs at 
\begin{equation}
q = q_* = \left[\frac{3(n+1)}{2 \alpha}\right]^{1/\alpha} \approx \left(\frac{3 n}{2 \alpha}\right)^{1/\alpha} 
\end{equation} 
if $n \gg 1$. The requirement that the worldline approximation is valid
implies that $q_* \alt 1/s$ and hence
\begin{equation}
n \alt \frac{2 \alpha}{3}\, s^{-\alpha} \,.
\label{eq:n-bound}
\end{equation} 
This condition gives the range of moments which are determined by the temporal sampling alone.
It is interesting to determine the interval of $x$ that largely determines these moments. 
If we use the approximation in Eq.~\eqref{eq:asymp-P} for
$P(x)$, the $n$-th moment is
\begin{equation}
\mu_n = \int_{-x_0}^\infty dx \, x^n\, P(x) \approx c_0 \int_0^\infty dx \, x^{n+b}\, {\rm e}^{-a x^c}\,.
\end{equation} 
The maximum of this integrand is at
\begin{equation}
x = x_n \approx \left(\frac{n}{a c}\right)^{1/c} \,,
\label{eq:xn-1}
\end{equation} 
if $n \gg b$. If we set $n$ equal to its upper limit in Eq.~\eqref{eq:n-bound}, then we obtain an estimate for the value of $x$ at which the transition from worldline 
to spacetime averaged behavior occurs:
\begin{equation}
x_* \approx s^{-3}  = (\tau/\ell)^3\,,
\end{equation} 
where we have used $c = \alpha/3$ and assumed that a factor of $a/2$ is of order one. As was discussed in Ref.~\cite{Huang:2016kmx}, $x \alt x_*$ is the 
range of validity of the worldline approximation. More generally  $x \approx x_*$ marks the transition in $P(x)$ from its worldline form to the spacetime averaged
form.

\subsection{The relative importance of different moments for the probability of large fluctuations}
\label{sec:relative}

We have seen that the lower moments, those which satisfy Eq.~\eqref{eq:n-bound}, determine the inner part of the probability distribution where $x \alt x_*$.
Similarly, we expect the higher moments to determine the region where $x \agt x_*$. We can make this statement more precise by noting that the form of
$P(x)$ for large $x$, given by either Eqs.~\eqref{eq:P2} or \eqref{eq:P3}, is also of the form of  Eq.~\eqref{eq:asymp-P} with $c = \alpha$. The argument leading
to Eq.~\eqref{eq:xn-1} still holds, and tells us that a given region of $P(x)$ for $x \agt x_*$ is determined by moments of order $n$, where
\begin{equation}
n \approx \alpha\, a\, x^\alpha\,.
\label{eq:n-of-x}
\end{equation} 
In this region, 
\begin{equation}
P(x) \propto  {\rm e}^{-a x^c} \approx {\rm e}^{-n/\alpha}\,.
\end{equation} 
This tells us that the value of $P(x)$ decreases exponentially with increasing $n$. 
The significance of this result lies in the fact that in a given application of the tail of probability distribution, we are typically interested in the probability of fluctuations
which might be large compared to the typical fluctuation, but for which $P(x)$ is still above some threshold of observability. Thus the regime of greatest physical
interest may be one where $x \gg 1$, but is not the $x \rightarrow \infty$ limit. 

Recall that the form of the tail of tail of $P(x)$ given by Eq.~\eqref{eq:P} was derived assuming that $S_N \approx T_N$ for large $N$. The numerical results
given in Figs.~\ref{fig:SNa},  \ref{fig:SNb}, and \ref{fig:SNc} indicate this happening in some cases. However, in other cases, especially the $\lambda = \alpha =1/2$
case in Fig.~\ref{fig:SNa}, $S_N$ is somewhat larger than $T_N$ for $N \alt 200$. Although the ratio $S_N/T_N$ is still decreasing, and might approach one
eventually, it is perhaps more important that  $S_N > T_N$ in many cases of physical interest. This implies that Eq.~\eqref{eq:P} is better viewed as a lower bound
on the actual probability distribution in these cases. For example, suppose that $S_N  \approx A\, T_N$ in some range of $N \gg 1$, where $A >1$ is a constant.
The corresponding range of $x$ is given by  Eq.~\eqref{eq:n-of-x}, given that $ n \approx N$ for $N \gg 1$. In this case, we can expect that Eq.~\eqref{eq:P}
underestimates the correct distribution in this range by a factor of $1/A$. Note that the overall constant in Eq.~\eqref{eq:P} is not determined by the arguments
presented in this paper. An alternative approach to computing  $P(x)$  is numerical diagonalization, which was used in  Ref.~\cite{SFF18} for the case of
time averaging. Work is currently in progress to extend this approach to the case of spacetime averaged operators. In principle, the diagonalization approach
is free of the ambiguities encountered in the present work.

\subsection{The case when the sampling length is large compared to the sampling time}
\label{sec:large-length}

In much of this paper, we have implicitly assumed that $s < 1$, or $\ell < \tau$. However, the opposite limit of large sampling length, $s >1$ is also of some interest. 
In this case, the diameter of the ball depicted in Fig.~\ref{fig:ball-and-shell}  is less than than the thickness of the shell. If $s \gg1$, the relevant illustration is the left-hand
panel of this figure, but with the ball entirely contained within the shell, as the case where the very small ball is partly outside the much thicker shell will give a small
contribution. In this case, the iteration will always be described by Eq.~\eqref{eq:largeq} with $C' = I(\infty)$, and the dominant contribution to the moments, $M_n$,
will be given by Eq.~\eqref{eq:Mnasy} with $C =1$ and $\mu =0$ for all $n$. However, the arguments in Sec.~\ref{sec:tail} that $S_N \approx T_N$ still require 
that $N \gg 1$. We may now write Eq.~\eqref{eq:P} for the asymptotic form of the tail of the probability distribution as
 \begin{equation}
 P(x) \approx K \, x^{1+2\lambda-4\alpha} \, {\rm e}^{-(x/B)^\alpha-\epsilon'\, (x/B)^\lambda} 
 \label{eq:P-largeL}
\end{equation} 
for $x \gg 1$, where the constant $K$ is found from Eqs.~\eqref{eq:Mnasy3} and \eqref{eq:B-def} to be
 \begin{equation}
 K =  \frac{3 C_{fg}^2}{32 \pi^4 \,\alpha^4}\,  2^{4-2(2+\lambda)/\alpha} \, B^{-2 (2+\lambda) +4\alpha} \,.
  \label{eq:K-largeL}
\end{equation} 
Unlike the more general case, here $K$ can be computed explicitly once the sampling functions are known.
Note that when $s >1$, Eq.~\eqref{eq:large-s} tells us that
 \begin{equation}
 B \approx \frac{B_1}{s^3} \,,
  \label{eq:B-largeL}
\end{equation} 
where $B_1$ is a constant. However, the factor of $C_{fg}^2$ is also a function of $s$.

Now we consider the special case where $\alpha = \lambda = 1/2$, where $\epsilon' = \epsilon = \sqrt{s} \gg 1$.  Now Eq.~\eqref{eq:P-largeL} becomes
 \begin{equation}
 P(x) \approx K  \, {\rm e}^{-\sqrt{s^4 \, x/B_1}}\,, 
 \label{eq:P-largeL2}
\end{equation} 
where
\begin{equation}
 K =  \frac{3 C_{fg}^2}{128 \pi^4\, B^3 } \,.
  \label{eq:K-largeL2}
\end{equation} 
Recall that $C_{fg} = C_f \, C_g$. Further assume that these constants have the values given in Sec.~\ref{sec:compact}: $C_f \approx 2.93$ and $C_g$
as given in Eq.~\eqref{eq:Cg}, and that $B_1 \approx 1$, as illustrated in Figs.~ \ref{fig:B1} and \ref{fig:B2}.  Finally, note that $s^4\, x = \ell^4\, T$, 
as $x = \tau^4 T$ and $T$ is the spacetime average of $:\dot{\varphi}^2:$. We may write the asymptotic probability distribution for $T$ as
\begin{equation}
P(T) \approx 1.5 \, s^6\, {\rm e}^{-\sqrt{\ell^4 \, T}} \,.
\label{eq:P-largeL3}
\end{equation} 
The factor of $s^6$ presumably reflects the fact that the limit $\tau \rightarrow 0$ for fixed $\ell$ is not meaningful. Equation~\eqref{eq:P-largeL3} is only valid 
when $T$ is sufficiently large that $P(T) \ll 1$.

\section{Summary and Discussion}
\label{sec:final}

In this paper, we have discussed the fluctuations of quantum stress tensor operators which have been averaged over finite intervals in both time and space. One can view
this spacetime averaging as modeling a measurement process which takes place in a finite spacetime region. Some averaging is essential for the operator
to have finite moments and hence a meaningful probability distribution. In the two spacetime dimensional CFT models treated in Sec.~\ref{sec:CFT}, the averaging
could be performed in time alone or equivalently in space alone, or it could be both in time and in space. In the latter case, the probability of large fluctuations is suppressed
compared to the cases of  time averaging alone or space averaging alone. In the four-dimensional models treated in the remainder of the paper, time averaging
is essential. Space averaging alone would not suppress an infinite contribution to the moments coming from pairs of modes associated with equal and opposite
momenta. For the same reason, there are no quantum inequalities for purely spatial averaging in four dimensions~\cite{FHR2002}.

 We have developed a formalism for treating the effects of both space and time averaging. In both cases, we assume that the averaging intervals are finite, 
and hence are described by compactly supported functions of time and of space. We have assumed that there is an inertial frame (a laboratory frame) in which
the space time averaging can be written as a product of a compactly supported function of time and of a spherically symmetric, compactly supported function of space.
The Fourier transform of the former is taken to be asymptotically proportional to ${\rm e}^{-|\omega\tau|^\alpha}$, and that of the latter to be asymptotically proportional to 
${\rm e}^{- (\ell k)^\lambda}$, where $0 < \lambda \leq \alpha <1$, $\tau$ is the characteristic width of the time
sampling functions, and $\ell$ is that of the spatial sampling function.

We developed an iteration procedure which generalizes that used in Ref.~\cite{FF2015} for the worldline case, and used this procedure to infer the rate of growth
of the moments and the asymptotic form of the stress tensor probability distribution, $P(x)$. Here $x = \tau^4\, T$ is a dimensionless measure of the averaged operator
$T$. We found that if the spatial sampling scale is small compared to the temporal scale, $\ell \ll \tau$, then there is finite range in $x$ which reproduces the worldline
result that $P(x) \sim c_0\, x^b \, {\rm e}^{-a x^c}$ with $c= \alpha/3$. However, as $x$ increases further, there is a transition region, beyond which $P(x)$ again
takes the same functional form, but with different values of the constants. We argued that the transition occurs at a value $x_*\approx (\tau/\ell)^3$. In particular, as 
$x \rightarrow \infty$, we find $c \approx \alpha$. This larger value of $c$ compared to the worldline case reflects the role of spatial averaging in suppressing large fluctuations. 
Nonetheless, with $\alpha <1$, the probability distribution
still falls more slowly than an exponential function. This allows the possibility of large physical effects from the fluctuations of space and time averaged stress tensors.

A typical vacuum fluctuation of the energy density or other stress tensor components is described by the root mean square value, $x_{\rm rms}$, which is expected to
be of order of one in $\tau =1$ units.  In the case where the switching function corresponds to $\alpha = 1/2$, then the probability density for a large fluctuation of the space and
time averaged energy density is roughly proportional to ${\rm e}^{-\sqrt{x}}$. A large fluctuation with $x = 100\, x_{\rm rms}$ is expected to be suppressed by a
factor of order ${\rm e}^{-10} = 4.5 \times 10^{-5}$ compared to a typical fluctuation. By comparison, in a process described by a Gaussian distribution, such a large
fluctuation would be suppressed by a factor of ${\rm e}^{-10^4}$.  

The results in this paper potentially have applications to several areas of physics, including phonon fluctuations in condensed matter physics, quantum tunneling, 
 density fluctuations in the early universe~\cite{WKF07,Ford:2010wd}, and the small scale structure of spacetime~\cite{CMP11,CMP18}.

\begin{acknowledgments} 
We would like to thank Peter Wu for useful discussions and comments on the manuscript.
This work was supported in part  by the National Science Foundation under Grant PHY-1607118, and by Scheme~4 Grant Ref.~41455 from the London Mathematical Society.
\end{acknowledgments}

 \appendix
 \section{Construction of an explicit choice of $\hat{f}(\omega)$ and of $\hat{g}(k)$}
 \label{sec:construct}
 
 In this appendix, we describe the construction of the specific forms of  $\hat{f}(\omega)$ and of $\hat{g}(k)$ which are used in the numerical computations
 reported in this paper. We first follow the procedure given in Sec.~IIB of \cite{FF2015}, and define the compactly supported function $H(t)$ by
 \begin{equation}
H(t) = \begin{cases} 
\frac{2}{\pi} (1-4t^2)^{-3/2} {\rm e}^{-1/(1-4t^2)}
& |t|<\frac{1}{2} \\ 0 & |t|\ge \frac{1}{2}\end{cases}\,.
\label{eq:H}
\end{equation} 
 Its Fourier transform is
 \begin{equation}
 \hat{H}(\omega) = \int_{-\infty}^\infty dt \, {\rm e}^{-i \omega t} \, H(t) = 2 \int_0^{1/2} dt \, \cos(\omega t) \, H(t) \,.
 \label{eq:Hhat}
 \end{equation} 
 
 In numerical computations, we avoid the singularity in the $(1-4t^2)^{-3/2}$ factor by setting the upper limit of integration to $0.499$.
 We define
 \begin{equation}
  \hat{L}(\omega) =  \hat{H}^2(\omega) +\frac{1}{2}[  \hat{H}^2(\omega+\pi) +   \hat{H}^2(\omega-\pi)] \,.
  \label{eq:Lhat}
 \end{equation} 
 Here the appearance of the square of $\hat{H}$ ensures that $ \hat{L}(\omega) \geq 0$, and the sum of three terms in Eq.~\eqref{eq:Lhat} is used
 to suppress oscillations as a function of $\omega$.
 Next let
 \begin{equation}
\hat{h}(\omega) = \frac{ \hat{L}(\omega)}{ \hat{L}(0)}\,.
\label{eq:hhat}
\end{equation} 
Now $\hat{h}(0) = 1$, so that $\hat{h}(\omega)$ is the Fourier transform of a normalized sampling function. Its asymptotic form for large arguments
is
\begin{equation}
\hat{h}_{\rm asy}(\omega) \approx 2.9324 {\rm e}^{-\sqrt{2 \omega}}\,.
\end{equation} 

It is useful to have a simple approximate form of $\hat{h}(\omega)$ for smaller values of its argument for use in numerical calculations. This can be
found by fitting a polynomial to numerically computed values for $\hat{h}(\omega)$, giving an approximation
\begin{eqnarray}
\hat{h}_{\rm fit}(\omega) &=& 1. - 0.0378271\, \omega^2 - 0.000429218\, \omega^3 + 0.000875262\, \omega^4  \\ \nonumber
 &-&  0.0000485667\, \omega^5 - 2.61062 \times 10^{-6} \, \omega^6 + 1.9601\times10^{-7}\, \omega^7\,, \qquad \omega < 9.92\,,
\end{eqnarray} 
and \,
 \begin{equation}
\hat{h}_{\rm fit}(\omega) = \hat{h}_{\rm asy}(\omega) \,, \qquad \omega \geq 9.92\,.
\end{equation} 
The value of $\omega = 9.92$, at which the polynomial  is matched to $\hat{h}_{\rm asy}(\omega)$ is selected to make the match as smooth as possible. 
The  function $\hat{h}(\omega)$, which is computed  using Eqs.~\eqref{eq:H}-\eqref{eq:hhat}, and its approximate form, $\hat{h}_{\rm fit}(\omega)$, are plotted
in Fig.~\ref{fig:hhat}. The matching region is illustrated in Fig.~\ref{fig:hfit}. 
 \begin{figure}[htbp]
	\centering
		\includegraphics[scale=0.3]{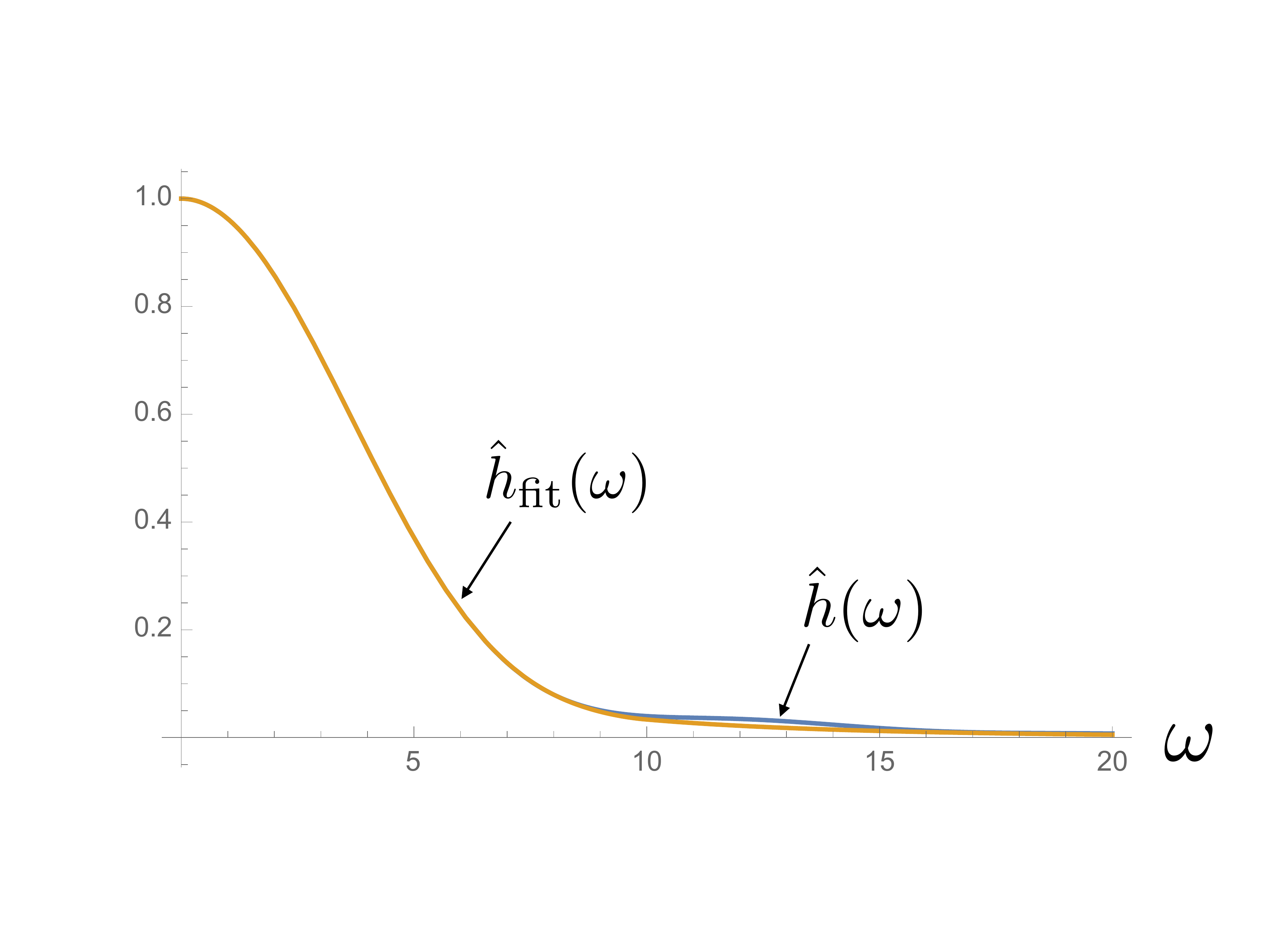}
		\caption{The functions $\hat{h}(\omega)$ and $\hat{h}_{\rm fit}(\omega)$  are illustrated. They are essentially identical on the scale shown,
		apart from a small local maximum in $\hat{h}(\omega)$  near $\omega =13$.}
			\label{fig:hhat}
\end{figure}
\begin{figure}[htbp]
	\centering
		\includegraphics[scale=0.3]{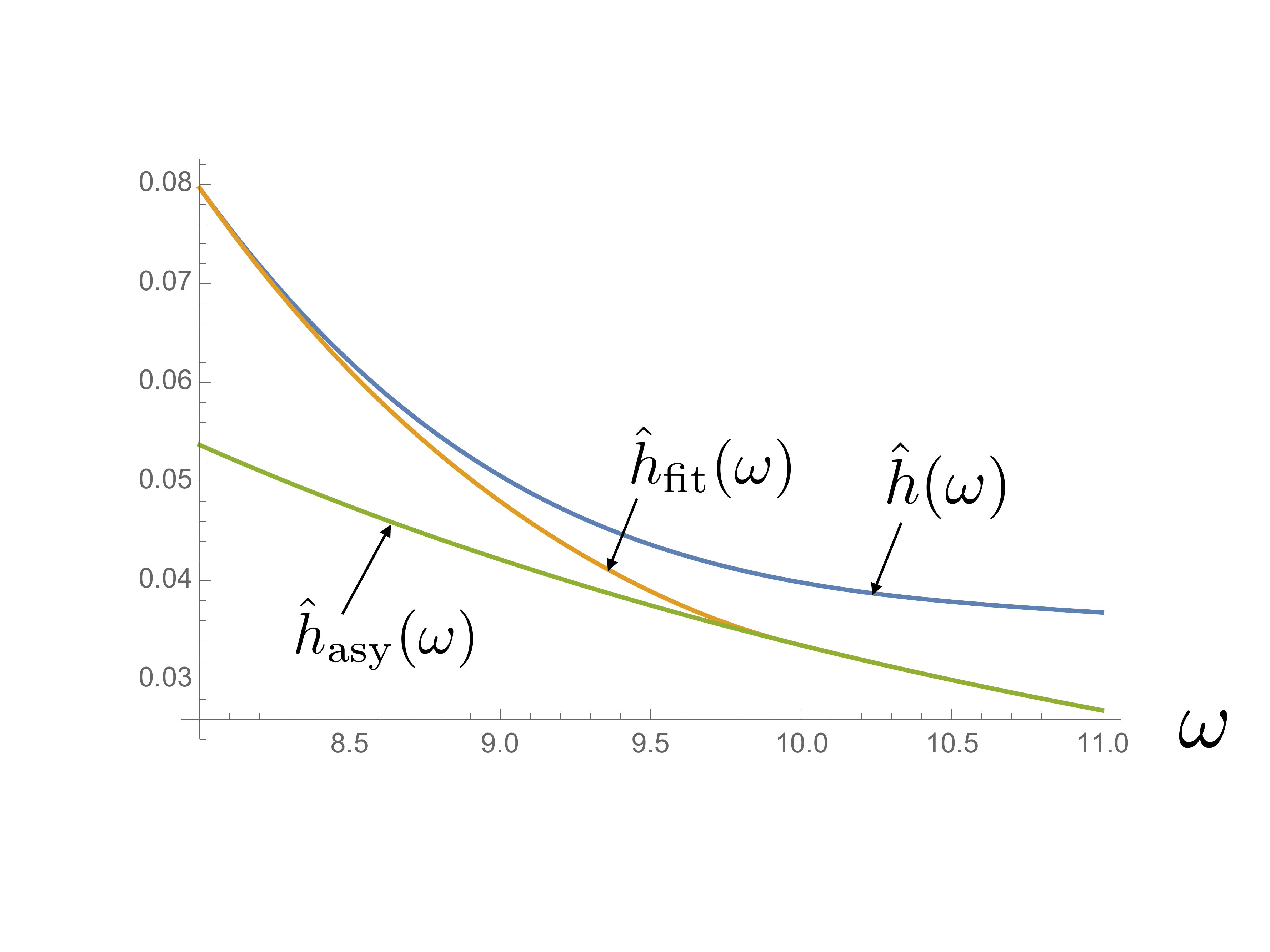}
		\caption{Here $\hat{h}(\omega)$ , its asymptotic form $\hat{h}_{\rm asy}(\omega)$ , and $\hat{h}_{\rm fit}(\omega)$  are illustrated
		near the matching region. The fitting function, $\hat{h}_{\rm fit}(\omega)$, has been chosen to interpolate as smoothly as possible
		between $\hat{h}(\omega)$ and  $\hat{h}_{\rm fit}(\omega)$. }
			\label{fig:hfit}
\end{figure}
For $\omega \leq 8$, the fractional error in the fit, $|\hat{h}_{\rm fit}(\omega) - \hat{h}(\omega)|/\hat{h}(\omega)$, is less than about $0.003$. For larger
values of $\omega$, $\hat{h}_{\rm fit}(\omega)$  was selected to approximate $\hat{h}_{\rm asy}(\omega)$. However, $\hat{h}(\omega)$ undergoes 
some oscillations before approaching $\hat{h}_{\rm asy}(\omega)$, as may be seen in Fig.~\ref{fig:hfit}. 

 We may use this choice of $\hat{h}_{\rm fit}(\omega)$ to define a temporal sampling function by $\hat{f}_{\rm fit}(\omega) = \hat{h}_{\rm fit}(\omega/2)$,
 and a spatial function, using Eq.~\eqref{eq:ghat}, by 
 \begin{equation}
\hat{g}_{\rm fit}(k) = \frac{\hat{h}'_{\rm fit}(k \ell)}{k \ell \, \hat{h}''_{\rm fit}(0)}\,.
\end{equation} 
The latter function is illustrated in Fig.~\ref{fig:gfit}.
\begin{figure}[htbp]
	\centering
		\includegraphics[scale=0.3]{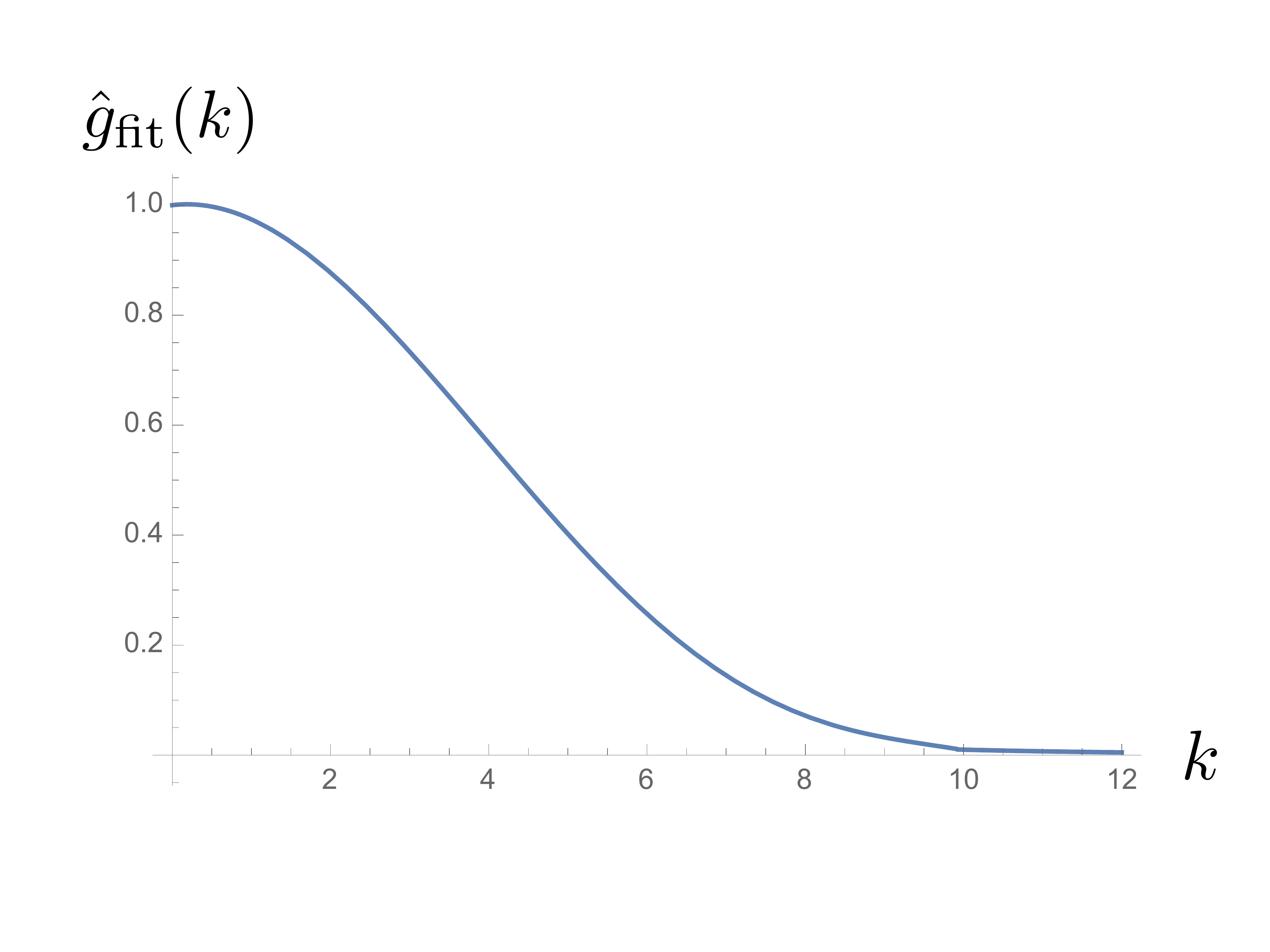}
		\caption{Here $\hat{g}_{\rm fit}(k)$ is plotted. It is the Fourier transform of the spherically symmetric spatial sampling function derived from
		$\hat{h}_{\rm fit}(\omega)$.}
			\label{fig:gfit}
\end{figure}

 \section{Fulks' generalization of Laplace's method}
 \label{sec:Fulks}
 
 The classical method of Laplace for asymptotic evaluation of integrals applies to 
 expressions of the form
 \begin{equation}
 I_h = \int_a^b f(t) e^{-h\phi(t)}\,dt
 \end{equation}
 as the parameter $h$ becomes large. As is well-known, the asymptotic behavior of $I_h$ is determined by the properties of $f$ and $h$ near the global minimum of 
 $\phi$ on the integration range, as well as the character of this minimum -- in particular, whether it is a stationary or nonstationary minimum, and whether it is located 
 at an endpoint or in the interior.  In this section we discuss more the general problem in which the integral 
 \begin{equation}\label{eq:Ihk}
 I_{h,k} = \int_a^b f(t) e^{-h\phi(t) + k\psi(t)}\,dt
 \end{equation}
 depends on two large parameters, both of which are becoming large, but at different rates. To be specific, we will assume that $k$ grows more slowly than $h$, to the 
 extent that $k=o(h)$ as $h\to\infty$.

 Fulks~\cite{Fulks:1951} considered integrals of the form~\eqref{eq:Ihk}
 where $-\infty<a<b\le\infty$, in which $\phi$ has a single global minimum at $a$. As he remarks, it is easy to generalize to the situation in which 
 $-\infty\le a<b\le\infty$ and $\phi$ has a single interior global minimum at $t_*\in (a,b)$, and we will state the results for this case.
 \begin{theorem}
 Suppose that
 \begin{itemize}
 	\item $\phi$ has a single global minimum at $t_*\in (a,b)$, near which it is $C^3$, and is nonincreasing in $[a,t_*]$ and nondecreasing in $[t_*,b]$
 	\item $\psi$ is $C^2$ near $t_*$, and continuous on $[a,b]$
 	\item $f$ is continuous at $t_*$ and $f(t_*)\neq 0$; it is also locally integrable and the integral $I_{h,k}$ exists for sufficiently large $h, k$.
 \end{itemize}
 Then if $h,k\to\infty$ with $k=o(h)$, the asymptotics may be given as follows:
 \begin{enumerate}
 	\item if $k=o(\sqrt{h})$ or $\psi'(t_*)=0$ then
 	\begin{equation}\label{eq:Fulks1}
 	I_{h,k}\sim f(t_*)\sqrt{\frac{2\pi}{h\phi''(t_*)}}\exp\left(-h\phi(t_*)+k\psi(t_*)\right);
 	\end{equation}
 	\item if $0<\liminf k/\sqrt{h}$ and $\limsup k/\sqrt{h}<\infty$ then
 	\begin{equation}\label{eq:Fulks2}
 	I_{h,k}\sim f(t_*)\sqrt{\frac{2\pi}{h\phi''(t_*)}}\exp\left(-h\phi(t_*)+k\psi(t_*)+\frac{\psi'(t_*)^2 k^2}{2\phi''(t_*)h}\right);
 	\end{equation}
 	\item if $\sqrt{h}=o(k)$ and $\psi'(t_*)\neq 0$ then
 	\begin{equation}\label{eq:Fulks3}
 	I_{h,k}\sim f(t_*)\sqrt{\frac{2\pi}{h\phi''(t_*)}}\exp\left(-h\phi(\tau)+k\psi(\tau)\right),
 	\end{equation}
 	where $\tau$ is determined by $h\phi'(\tau)=k\psi'(\tau)$ and is the position of the global minimum of $-h\phi(t)+k\psi(t)$. If, more specifically, $k=o(h^{2/3})$, one has 
 	\begin{equation}\label{eq:Fulks4}
 	I_{h,k}\sim f(t_*)\sqrt{\frac{2\pi}{h\phi''(t_*)}}\exp\left(-h\phi(t_*)+k\psi(t_*)+\frac{\psi'(t_*)^2 k^2}{2\phi''(t_*)h}\right).
 	\end{equation}
 	(Other special cases can be given, for different conditions on the growth of $k$ relative to $h$ and suitable higher regularity of $\phi$ and $\psi$. In general we can solve for $\tau$ as a series in $k/h$ and 
 	the exponent will contain terms proportional to $h(k/h)^a$ for all $a\in\NN_0$ so that $h(k/h)^a$ is constant or growing as $h\to\infty$).
 \end{enumerate}
\end{theorem}
\begin{proof}
	Apart from the parenthetic comment, all the statements are lightly adapted from Theorems~1--4 and the Corollary of~\cite{Fulks:1951}, noting the comments that follow the Corollary. The comment is evident 
	by expanding the inverse function to $\eta(t)=\phi'(t)/\psi'(t)$ using Taylor's theorem with remainder, noting that $\tau=\eta^{-1}(k/h)$.
\end{proof}

 As an example, we consider the integrals
 \begin{equation}
 I_N = \int_0^\infty dq\, q^{N-1} e^{-q^\alpha-\epsilon q^\lambda},
 \label{eq:IN}
 \end{equation}
 where $0<\lambda<\alpha<1$, defined in Eq.~\eqref{eq:INtext}. [For reference, the case $\lambda=\alpha$ can be evaluated exactly to give $I_N = \alpha^{-1} \Gamma(N/\alpha) (1+\epsilon)^{-N/\alpha}$.]
 Changing variables to $v=q^\alpha$ gives
 \begin{equation}
 I_N = \alpha^{-1}\int_0^\infty dv\, v^{N/\alpha -1} e^{-v-\epsilon v^{\beta}} ,
 \end{equation}
 in which the integral is known as Fax\'en's integral,  $I_N=\alpha^{-1}\mathrm{Fi}(\beta,N/\alpha;-\epsilon)$ in the notation of  \cite[\S 9.4]{Olver:1974}. 
 Setting $\Omega=N/\alpha - 1$ and $\beta=\lambda/\alpha$, and making the change of variables $v=\Omega t$, we have
 \begin{equation}
 I_N = \frac{\Omega^{\Omega+1}}{\alpha}\int_0^\infty dt\, e^{\Omega(\log t-t)-\epsilon \Omega^{\beta}t^{\beta}},
 \end{equation}
 in which the integral is of Fulks' form with $h=\Omega$, $k=\Omega^\beta$, $\phi(t)=t-\log t$, $\psi(t)=-\epsilon t^\beta$, $f\equiv 1$. Noting that 
 \begin{equation}
 \phi'(t)=1-t^{-1},\quad \phi''(t)=t^{-2}
 \end{equation}
 we see that $\phi$ has a single global minimum at $t_*=1$, to the left of which it is decreasing and to the right of which it is increasing. 
 Note that $\phi(t_*)=\phi''(t_*)=1$, $\psi(t_*)=-\epsilon$, $\psi'(t_*)=-\beta\epsilon$. There are several cases, 
 depending on the value of $\beta=\lambda/\alpha$:
 \begin{itemize}
 	\item if $\lambda<\alpha/2$, then $k=o(\sqrt{h})$ and by~\eqref{eq:Fulks1},
 	\begin{equation}
 	I_N \sim  \frac{\Omega^{\Omega+1/2}e^{-\Omega-\epsilon\Omega^{\lambda/\alpha}}\sqrt{2\pi}}{\alpha}
 	\sim \alpha^{-1}\Gamma(N/\alpha) e^{-\epsilon(N/\alpha-1)^{\lambda/\alpha}};
	\label{eq:Lt-half}
 	\end{equation}
 	\item if $\lambda=\alpha/2$, then $k=\sqrt{h}$ and by~\eqref{eq:Fulks2},
 	\begin{equation}
 	I_N \sim  \frac{\Omega^{\Omega+1/2}e^{-\Omega-\epsilon\Omega^{\lambda/\alpha}+\epsilon^2/8}\sqrt{2\pi}}{\alpha}
 	\sim \alpha^{-1}e^{\epsilon^2/8}\Gamma(N/\alpha) e^{-\epsilon\sqrt{N/\alpha-1}};
	\label{eq:eq-half}
 	\end{equation}
 	\item if $\alpha/2<\lambda<2\alpha/3$ then $\sqrt{h}=o(k)$ and $k=o(h^{2/3})$, and by~\eqref{eq:Fulks4},
 	\begin{equation}
 	I_N \sim  \frac{\Omega^{\Omega+1/2}e^{-\Omega-\epsilon\Omega^{\lambda/\alpha} +(\epsilon\beta)^2\Omega^{2\beta-1}/2}\sqrt{2\pi}}{\alpha}
 	\sim \alpha^{-1} \Gamma(N/\alpha) e^{-\epsilon(N/\alpha-1)^{\lambda/\alpha}
 		+(\epsilon\beta)^2(N/\alpha-1)^{2\beta-1}/2}
		\label{eq:Lt-two-thirds}
 	\end{equation} 
 	(so this formula also holds for $\lambda=\alpha/2$);
 	\item if $\alpha/2<\lambda<\alpha<1$, with no further information, then by~\eqref{eq:Fulks3},
 	\begin{equation}
 	I_N \sim  \frac{\Omega^{\Omega+1/2}\tau^\Omega e^{-\Omega\tau-\epsilon(\Omega\tau)^{\lambda/\alpha}}\sqrt{2\pi}}{\alpha} \sim
 	\alpha^{-1}\Gamma(N/\alpha)\tau^{N/\alpha-1} e^{-(N/\alpha-1)(\tau-1)-\epsilon((N/\alpha-1)\tau)^{\lambda/\alpha}},
 	\end{equation}
 	where $\tau$ is determined by $\tau+\epsilon\beta\Omega^{\beta-1}\tau^\beta=1$, $\beta=\lambda/\alpha$.
 	More could be said given a tighter upper bound on $\lambda/\alpha$ and indeed the formula~\eqref{eq:Lt-two-thirds} given for the range $\lambda\in (\alpha/2,2\alpha/3)$ is a special case.
 \end{itemize}

 As a check on the result for $\lambda=\alpha/2$, we note that $I_N$ can be evaluated in terms of Kummer functions in this case.
 Changing variables to $v=q^{\alpha/2}$, one has
 \begin{equation}
 I_N = \frac{2}{\alpha} \int_0^\infty dv\, v^{2N/\alpha-1}e^{-v^2-\epsilon v},
 \end{equation}
 which evaluates by~\cite[3.462.1]{GR2000}
 \begin{align}
 I_N &= \frac{1}{\alpha}2^{1-N/\alpha}\Gamma(2N/\alpha)e^{\epsilon^2/8}D_{-2N/\alpha}(\epsilon/\sqrt{2})  \\
 &=
\frac{2}{\alpha}2^{-2N/\alpha}\Gamma(2N/\alpha) U\left(\frac{N}{\alpha},\frac{1}{2}, \frac{\epsilon^2}{4}\right) \\
&= \frac{\Gamma(N/\alpha) \Gamma(N/\alpha+1/2)}{\alpha\sqrt{\pi}}   U\left(\frac{N}{\alpha},\frac{1}{2}, \frac{\epsilon^2}{4}\right). 
\end{align}
where $D_{\nu}(z)$ is a parabolic cylinder function and $U$ is Kummer's function (see~\cite[\S 12.1 \& 12.7.14]{DLMF} for the relation between these special functions).
The last step uses the duplication formula for $\Gamma$-functions.  
 
Asymptotic expansions of the Kummer function $U$ for large parameters are known -- see~\cite[\S13]{DLMF} and~\cite{Temme:2013} -- and give
 \begin{equation}
 I_N(\epsilon) \sim  \frac{\Gamma(N/\alpha+1/2)e^{\epsilon^2/8}}{\alpha \sqrt{N/\alpha-1}}e^{-\epsilon\sqrt{N/\alpha-1}}\sim  \frac{\Gamma(N/\alpha)e^{\epsilon^2/8}}{\alpha }e^{-\epsilon\sqrt{N/\alpha-1}}
 \end{equation}
 in agreement with our results above.


\begin{thebibliography}{00}

 
 \bibitem{FewsterFordRoman:2010}
 C.J.~Fewster, L.H.~Ford and T.A.~Roman, ``Probability distributions of smeared
 quantum stress tensors,''
 Phys. Rev. D \textbf{81}, {121}{901} (2010), arXiv:1004.0179 [quant-ph].
 
 \bibitem{FFR2012} 
  C.~J.~Fewster, L.~H.~Ford and T.~A.~Roman,
  ``Probability distributions for quantum stress tensors in four dimensions,''
  Phys.\ Rev.\ D {\bf 85}, 125038 (2012),
  arXiv:1204.3570 [quant-ph].
  
 \bibitem{FF2015} 
  C.~J.~Fewster and L.~H.~Ford,
  ``Probability Distributions for Quantum Stress Tensors Measured in a Finite Time Interval,''
  Phys.\ Rev.\ D {\bf 92}, 105008 (2015),
  arXiv:1508.02359 [hep-th]. 

\bibitem{Fe&Ho18}
C.J.~Fewster and S.~Hollands, ``Probability distributions for the stress tensor in conformal field theories,''
Lett.\ Math.\ Phys. {\bf 109} 747--780 (2019), arXiv:1805.04281 [gr-qc].

 
\bibitem{F2017}
C.J.~Fewster, ``Quantum Energy Inequalities'' 
in {\it Wormholes, Warp Drives and Energy Conditions}, edited by FSN Lobo. Fundamental Theories of Physics, vol 189. (Springer, Cham, 2017).

\bibitem{Simon} B. Simon, ``The classical moment problem as a self-adjoint finite difference operator",
Adv. Math. {\bf 137}, 82 (1998).
  
  \bibitem{SFF18}  E. D. Schiappacasse, C. J. Fewster and L. H. Ford, ``Vacuum Quantum Stress Tensor Fluctuations: A Diagonalization Approach,''  
 Phys. Rev. D  {\bf 97}, 025013 (2018), 
 arXiv:1711.09477 [hep-th].

\bibitem{Huang:2016kmx} 
  H.~Huang and L.~H.~Ford,
  ``Vacuum Radiation Pressure Fluctuations and Barrier Penetration,''
  Phys.\ Rev.\ D {\bf 96}, 016003 (2017),
  arXiv:1610.01252 [quant-ph].

 \bibitem{FZ99}    V.V. Flambaum and V.G. Zelevinsky, 
 ``Radiation Corrections Increase Tunneling Probability",
 Phys. Rev. Lett {\bf 83}, 3108 (1999).

\bibitem{Huang:2015lea} 
  H.~Huang and L.~H.~Ford,
  ``Quantum Electric Field Fluctuations and Potential Scattering,''
  Phys.\ Rev.\ D {\bf 91}, 125005 (2015),
  arXiv:1503.02962 [hep-th].
  
  \bibitem{WKF07} C.-H. Wu, K.-W. Ng, and L.H. Ford,  
  "Possible constraints on the duration of inflationary expansion from quantum stress tensor fluctuations",
  Phys. Rev. D {\bf 75}, 103502 (2007), arXiv:gr-qc/0608002.
  
  \bibitem{Ford:2010wd} 
  L.~H.~Ford, S.~P.~Miao, K.~W.~Ng, R.~P.~Woodard and C.~H.~Wu,
  ``Quantum Stress Tensor Fluctuations of a Conformal Field and Inflationary Cosmology,''
  Phys.\ Rev.\ D {\bf 82}, 043501 (2010),
  arXiv:1005.4530 [gr-qc].
  
  \bibitem{Wu:2011gk} 
  C.~H.~Wu, J.~T.~Hsiang, L.~H.~Ford and K.~W.~Ng,
  ``Gravity Waves from Quantum Stress Tensor Fluctuations in Inflation,''
  Phys.\ Rev.\ D {\bf 84}, 103515 (2011),
  arXiv:1105.1155 [gr-qc].
 
  \bibitem{CMP11} S. Carlip, R.A. Mosna and J.P.M. Pitelli, "Vacuum Fluctuations and the small
scale structure of spacetime", Phys. Rev. Lett. 107, 021303 (2011), arXiv:1103.5993  [hep-th].
   
 \bibitem{CMP18}   S. Carlip, Ricardo A. Mosna, J. P. M. Pitelli,
"Quantum Fields, Geometric Fluctuations, and the Structure of Spacetime",
arXiv:1809.08265  [gr-qc].
 
 \bibitem{AF19} M. C. Anthony and C. J. Fewster, "Explicit examples of probability distributions for the energy density in
two-dimensional conformal field theory", arXiv:1908.00393 [hep-th].

\bibitem{Fe&Ho05}
C.J.~Fewster and S.~Hollands, ``Quantum energy inequalities in two-dimensional
conformal field theory,''
\newblock Rev. Math. Phys. \textbf{17}, 577--612 (2005), arXiv:math-ph/0412028.

\bibitem{FHR2002} L.~H.~Ford, A.~D.~Helfer, T.~A.~Roman,
``Spatially Averaged Quantum Inequalities Do Not Exist in Four-Dimensional Spacetime,'' Phys.\ Rev.\ D {\bf 66} (2002) 124012, arXiv:gr-qc/0208045.

\bibitem{Fulks:1951} W.~Fulks, 
	``A generalization of Laplace's method,'' 
	Proc.\ Am.\ Math.\ Soc.\ {\bf 2}, No. 4, 613--622 (1951).

\bibitem{Olver:1974}
F.~W.~J.~Olver,
\textit{Asymptotics and special functions}, (Academic Press, 1974). 

\bibitem{GR2000}
I.~S.~Gradshteyn and I.~M.~Ryzhik, \textit{Table of Integrals, Series, and Products}, 6th ed., (Academic Press, 2000).

\bibitem{DLMF} F. W. J. Olver, A. B. Olde Daalhuis, D. W. Lozier, B. I. Schneider, R. F. Boisvert, C. W. Clark, B. R. Miller, and B. V. Saunders, eds.,
\textit{NIST Digital Library of Mathematical Functions} \url{http://dlmf.nist.gov/}, Release 1.0.20 of 2018-09-15.  

\bibitem{Temme:2013}
N.~M.~Temme, 
``Remarks on Slater's asymptotic expansions of Kummer functions for large values of the $\alpha$-parameter.''
Adv. Dyn. Syst. Appl. {\bf 8}, 365--377 (2013). 


\end{thebibliography}
\end{document}